\documentclass{fundam}

\publyear{22}
\papernumber{2102}
\volume{185}
\issue{1}

\usepackage[auth-sc]{authblk}

\usepackage[utf8]{inputenc}

\usepackage{verbatim}
\usepackage{graphicx}
\usepackage{rotating}

\usepackage{adjustbox}

\usepackage{tikz}
\usepackage{tikz-cd}
\usetikzlibrary{backgrounds,circuits,circuits.ee.IEC,shapes,fit,matrix}

\tikzstyle{simple}=[-,line width=2.000]
\tikzstyle{arrow}=[-,postaction={decorate},decoration={markings,mark=at position .5 with {\arrow{>}}},line width=1.100]
\pgfdeclarelayer{edgelayer}
\pgfdeclarelayer{nodelayer}
\pgfsetlayers{edgelayer,nodelayer,main}

\tikzstyle{none}=[inner sep=0pt]

\definecolor{lblue}{rgb}{0,250,255}
\tikzstyle{species}=[circle,fill=yellow,draw=black,scale=1.15]
\tikzstyle{transition}=[rectangle,fill=lblue,draw=black,scale=1.15]
\tikzstyle{inarrow}=[->, >=stealth, shorten >=.03cm,line width=1.5]
\tikzstyle{empty}=[circle,fill=none, draw=none]
\tikzstyle{inputdot}=[circle,fill=purple,draw=purple, scale=.25]
\tikzstyle{inputarrow}=[->,draw=purple, shorten >=.05cm]
\tikzstyle{simple}=[-,draw=purple,line width=1.000]

\usepackage{url}
\usepackage{xcolor}
\definecolor{linkColor}{RGB}{156,78,13}
\usepackage{hyperref}
\hypersetup{
    colorlinks = true,
    allcolors = linkColor 
    }
\usepackage{xspace}
\usepackage{pslatex}

\usepackage{algorithm}
\usepackage{algorithmic}

\usepackage{datetime}
\newdateformat{versiondate}{%
\THEMONTH\THEDAY}

\usepackage{mathtools}

\usepackage{color}

\usepackage[inline]{enumitem}
\usepackage{amsmath}
\usepackage{framed} 

\usepackage{listings}
\usepackage{jlcode} 

\DeclareUnicodeCharacter{03B4}{$\delta$}
\DeclareUnicodeCharacter{2218}{$\circ$}
\DeclareUnicodeCharacter{222B}{$\int$}
\DeclareUnicodeCharacter{2080}{$_0$}
\DeclareUnicodeCharacter{2081}{$_1$}
\DeclareUnicodeCharacter{2208}{$\in$}

\usepackage{sectsty}
\allsectionsfont{\fontfamily{lmss}\selectfont}

\ifdefined\DEBUG{}

\else

\fi

\newcommand{\id}{\mathsf{id}}

\usepackage{pgf,tikz,environ}
\usepackage{quiver}
\usetikzlibrary{arrows}
\usetikzlibrary{cd}
\usetikzlibrary{positioning}
\usetikzlibrary{fadings}
\usetikzlibrary{decorations.pathreplacing}
\usetikzlibrary{decorations.pathmorphing}
\tikzset{snake it/.style={decorate, decoration=snake}}

\definecolor{parchment}{RGB}{214, 204, 169}
\theoremstyle{definition}
\newtheorem{problem}[section]{Open Problem}

\newcommand{\Oh}{\mathcal{O}}

\newcommand{\define}[1]{\textbf{#1}}

\newcommand{\cat}[1]{\mathsf{#1}}
\newcommand{\typesetcategory}[1]{\cat{#1}}

\newcommand{\typesetoperator}[1]{\mathsf{#1}}

\DeclareMathOperator{\sieve}{\typesetoperator{Sieve}}

\DeclareMathOperator{\tw}{\typesetoperator{tw}}

\DeclareMathOperator{\domain}{\typesetoperator{dom}}
\DeclareMathOperator{\codomain}{\typesetoperator{cod}}

\DeclareMathOperator{\morphisms}{\typesetoperator{Mor}}
\DeclareMathOperator{\image}{\typesetoperator{Im}}

\DeclareMathOperator{\CAT}{\typesetcategory{Cat}}
\DeclareMathOperator{\finset}{\typesetcategory{FinSet}}

\DeclareMathOperator{\set}{\typesetcategory{Set}}

\DeclareMathOperator{\fingr}{\typesetcategory{FinGr}}
\DeclareMathOperator{\GRhomo}{\typesetcategory{Gr}_{H}}

\DeclareMathOperator{\sd}{\mathfrak{D}_m}
\DeclareMathOperator{\const}{\typesetoperator{const}}
\DeclareMathOperator{\colim}{\typesetoperator{colim}}
\DeclareMathOperator{\decision}{\typesetoperator{dec}}

\DeclareMathOperator*{\submon}{\typesetcategory{subMon}}
\DeclareMathOperator*{\decomp}{\typesetcategory{Dcmp}}

\newcommand{\typesetproblem}[1]{\textsc{#1}}

\DeclareMathOperator{\fpt}{\typesetproblem{FPT}}

\DeclareMathOperator{\Aa}{\mathcal{A}}

\DeclareMathOperator{\Fa}{\mathcal{F}}  
\DeclareMathOperator{\Ga}{\mathcal{G}}

\DeclareMathOperator{\Ta}{\mathcal{T}}

\begin{document}

\title{Compositional Algorithms on Compositional Data: Deciding Sheaves on Presheaves}

\address{Computer \& Information Science \& Engineering, 
University of Florida,
432 Newell Drive, Gainesville, FL 32603, USA.}

\author{
Ernst Althaus \\ 
Institute of Computer Science \\
Johannes Gutenberg-University \\ 
Staudingerweg 9 55128, Mainz, Germany \\
ernst.althaus{@}uni-mainz.de 
\and 
Benjamin Merlin Bumpus\thanks{DARPA ASKEM and Young Faculty Award programs through grants HR00112220038 and W911NF2110323}\\ 
Computer and Information Science and Engineering \\ 
University of Florida \\ 
432 Newell Drive, Gainesville, FL 32603, USA. \\ benjamin.merlin.bumpus{@}gmail.com 
\and 
 James Fairbanks\thanks{DARPA ASKEM and Young Faculty Award programs through grants HR00112220038 and W911NF2110323} \\ 
Computer and Information Science and Engineering \\ 
University of Florida \\ 
432 Newell Drive, Gainesville, FL 32603, USA. \\ 
fairbanksj{@}ufl.edu 
\and 
Daniel Rosiak \\ 
National Institute of Standards and Technology \\ 
100 Bureau Dr, Gaithersburg, MD 20899, USA \\ 
danielhrosiak{@}gmail.com
}

\maketitle
\runninghead{Althaus, Bumpus, Fairbanks, Rosiak}{Deciding Sheaves on Presheaves}

\begin{abstract}
Algorithmicists are well-aware that fast dynamic programming algorithms are very often the correct choice when computing on compositional (or even recursive) \textit{graphs}. Here we initiate the study of how to generalize this folklore intuition to mathematical structures writ large. We achieve this horizontal generality by adopting a categorial perspective which allows us to show that: (1) \textit{structured decompositions} (a recent, abstract generalization of many graph decompositions) define Grothendieck topologies on categories of data (adhesive categories) and that (2) any computational problem which can be represented as a \textit{sheaf} with respect to these topologies can be decided in linear time on classes of inputs which admit decompositions of bounded width and whose decomposition shapes have bounded feedback vertex number. This immediately leads to algorithms on objects of any \(\cat{C}\)-set category; these include -- to name but a few examples -- structures such as: symmetric graphs, directed graphs, directed multigraphs, hypergraphs, directed hypergraphs, databases, simplicial complexes, circular port graphs and half-edge graphs. 

Thus we initiate the bridging of tools from sheaf theory, structural graph theory and parameterized complexity theory; we believe this to be a very fruitful approach for a general, algebraic theory of dynamic programming algorithms. Finally we pair our theoretical results with concrete implementations of our main algorithmic contribution in the \textit{AlgebraicJulia} ecosystem.
\end{abstract}

\begin{keywords}
Parameterized Complexity, Dynamic Programming, Sheaf Theory, Category Theory
\end{keywords}

\section{Introduction}\label{sec:introduction}
As pointed out by Abramsky and Shah~\cite{structure-and-power}, there are two main ``organizing principles in the foundation of computation'': these are \textbf{structure} and \textbf{power}. The first is concerned with compositionality and semantics while the second focuses on expressiveness and computational complexity. 

So far these two areas have remained largely disjoint. This is due in part to mathematical and linguistic differences, but also to perceived differences in the research focus. However, we maintain that many of these differences are only superficial ones and that \textit{compositionality} has always been a major focus of the ``power'' community. Indeed, although this has not yet formalized as such (until now), \textit{dynamic programming} and \textit{graph decompositions} are clear witnesses of the importance of compositionality in the field. Thus the main characters of the present story are:  
\begin{enumerate}
    \item[\((\mathbf{SC})\)] the \textbf{Structural} compositionality arising in graph theory in the form of \textit{graph decompositions} and graph \textit{width measures} (whereby one decomposes graphs into smaller and simpler constituent parts~\cite{diestel-graph-decomps-infinite, Diestel2010GraphTheory, RobertsonII, robertsonX, Bertele1972NonserialProgramming, halin1976s, oumThesis, GEELEN2020, RobertsonXX, WOLLAN201547, dtw, berwanger2012dag, hunter2008digraph, safari2005d, kreutzer2018, bumpus2020directed, CARMESIN2022101, layered-tw-dujmovic2017, layered-tw-shahrokhi2015, H-tw-conference}),
    \item[\((\mathbf{AC})\)] the \textbf{Algorithmic} compositionality embodied by the intricate \textit{dynamic programming} routines found in parameterized complexity theory~\cite{GroheBook, CourcelleBook, flum2006parameterized, cygan2015parameterized, downey2013fundamentals}, 
    \item[\((\mathbf{RC})\)] and the \textbf{Representational} compositionality\footnote{We choose the term ``representational'' as a nod to Leray's~\cite{leray1946lanneau} 1946 understanding of what sheaf theory should be: ``Nous nous proposons d'indiquer sommairement comment les m\'ethodes par lesquelles nous avons etudi\'e la topologie d'un espace peuvent \^etre adapt\'ees \`a l'\'etude de la topologie d'une repr\'esentation'' (in English: ``We propose to indicate briefly how the methods by which we have studied the topology of a space can be adapted to the study of the topology of a representation''). As Gray~\cite{gray2006fragments} points out, this is ``the first place in which the word `faisceau' is used with anything like its current mathematical meaning''.} arising in algebraic topology and virtually throughout the rest of mathematics in the form of \textit{sheaves}~\cite{gray2006fragments, rosiak-book, maclane2012sheaves} (whereby one systematically assigns data to `spaces' in such a way that one can easily keep track of how local data interacts with global data).
\end{enumerate}

The study of graph decompositions and their associated `width measures' has been an extremely active area of research which has generated a myriad of subtly different methods of decomposition~\cite{Bertele1972NonserialProgramming, halin1976s, RobertsonII, CourcelleBook, GroheBook, diestel-graph-decomps-infinite, oum2006approximating, dtw, berwanger2012dag, hunter2008digraph, safari2005d, kreutzer2018, CARMESIN2022101, layered-tw-dujmovic2017, layered-tw-shahrokhi2015, H-tw-conference, WOLLAN201547, bumpus2020directed, bumpus2022edge}. From a computationally minded perspective, these notions are crucial for dynamic programming since they can be seen as data structures for \textit{graphs} with compositional structure~\cite{GroheBook, CourcelleBook, cygan2015parameterized, downey2013fundamentals}. 

Sheaves on the other hand supply the canonical mathematical structure for attaching data to spaces (or something ``space-like''), where this further consists of restriction maps encoding some sort of local constraints or relationships between the data, and where the sheaf structure prescribes how compatible local data can be combined to supply global structure. As such, the sheaf structure lets us reason about situations where we are interested in tracking how compatible local data can be stitched together into a global data assignment. While sheaves were already becoming an established instrument of general application by the 1950s, first playing decisive roles in algebraic topology, complex analysis in several variables, and algebraic geometry, their power as a framework for handling all sorts of local-to-global problems has made them useful in a variety of areas and applications, including sensor networks, target tracking, dynamical systems (see the textbook by Rosiak~\cite{rosiak-book} for many more examples and references).

\paragraph{Our contribution} is to show how to use tools from both sheaf theory and category theory to bridge the chasm separating the ``structure'' and ``power'' communities. Indeed our main \textit{algorithmic} result (Theorem~\ref{thm:filtering-algorithm}) is a meta-theorem obtained by amalgamating these three perspectives. Roughly it states that any decision problem which can be formulated as a \textit{sheaf} can be solved in \textit{linear time} on classes of inputs which can built compositionally out of constituent pieces which have \textit{bounded internal complexity}. In more concrete terms, this result yields linear (\(\fpt\)-time) algorithms for problems such as \(H\)-coloring (and generalizations thereof) on a remarkably large class of mathematical structures; a few examples of these are: \begin{enumerate*} \item databases,\item simple graphs, \item directed graphs, \item directed multigraphs, \item hypergraphs, \item directed hypergraphs (i.e. Petri nets), \item simplicial complexes, \item circular port graphs~\cite{libkind2021operadic} and \item half-edge graphs \end{enumerate*}.

To be able to establish (or even state) our main algorithmic result, one must first overcome two highly nontrivial hurdles. The first is that one needs to make sense of what it means for mathematical objects (which need not be graphs) to display compositional structure. The second is that, in order to encode computational problems as sheaves, one must first equip the input objects with a notion of a \textit{topology} which moreover is expressive enough to encode the compositional structure of the input instances. 

We overcome the first issue by abandoning the narrow notion of \textit{graph} decompositions in favor of its broader, object-agnostic counterpart, namely \textit{structured decompositions}. These are a recent category-theoretic generalization of graph decompositions to objects of arbitrary categories~\cite{structured-decompositions}. 

This shift from a graph-theoretic perspective to a category-theoretic one also allows us to overcome the second hurdle by thinking of decompositions as \textit{covers} of objects in a category. More precisely we prove that structured decompositions equip categories of data (i.e. adhesive categories) with  \textit{Grothendieck topologies}. This is our main \textit{structural} result (Corollary~\ref{corollary:decomps-as-topologies}). It empowers us with with the ability to think of \textit{sheaves} with respect to the \textit{decomposition topology} as a formalization of the vague notion of `computational problems whose compositional structure is compatible with that of their inputs'. 

\subsection{Roadmap}
Throughout the paper we will built up all the necessary category theoretic formulations to make the following illustration a precise and formal commutative diagram (this will be Diagram~\ref{diagram:algorithmic-goal-NO-COLIMITS}). 
\begin{equation}\label{diagram:ENGLISH}
\adjustbox{scale=1.5, max width=\textwidth}{
\begin{tikzcd}[row sep=huge]
	& {\text{Data}} &&& {\text{Sol. Spaces}} &&&&& {\text{Answer Space}} \\
	{\text{Data w/ Decomps.}} \\
	& {\text{Decomps. of Data}} &&& {\text{Decomps. of Sol. Spaces}} && {\text{Decomps. of Sol. Spaces}} &&& {\text{Decomps. of Answers}}
	\arrow[""{name=0, anchor=center, inner sep=0}, "{\text{collect}}"', color={rgb,255:red,214;green,92;blue,214}, from=3-5, to=1-5]
	\arrow["{\text{GlobalSolSpace}}", color={rgb,255:red,92;green,92;blue,214}, from=1-2, to=1-5]
	\arrow["{\text{LocalSolSpace}}"', color={rgb,255:red,214;green,92;blue,92}, from=3-2, to=3-5]
	\arrow["{comm.}", shift right=5, draw=none, from=3-2, to=1-5]
	\arrow["{\text{is Empty?}}", color={rgb,255:red,92;green,92;blue,214}, from=1-5, to=1-10]
	\arrow["{\text{Algorithm}}"', color={rgb,255:red,214;green,92;blue,92}, from=3-5, to=3-7]
	\arrow["{\text{is Empty (locally) ?}}"', color={rgb,255:red,214;green,92;blue,92}, from=3-7, to=3-10]
	\arrow[""{name=1, anchor=center, inner sep=0}, "\land"', color={rgb,255:red,214;green,92;blue,92}, from=3-10, to=1-10]
	\arrow[color={rgb,255:red,92;green,92;blue,214}, from=2-1, to=1-2]
	\arrow[color={rgb,255:red,214;green,92;blue,92}, from=2-1, to=3-2]
	\arrow["{comm.}"{description}, draw=none, from=0, to=1]
\end{tikzcd}
}
\end{equation}
This diagram, whose arrows will be functors relating appropriate categories, relates \textbf{structure} and \textbf{power}  by succinctly encoding all three forms of compositionality we identified earlier. The categories on the bottom row correspond to \textbf{Structural Compositionality}. The colored paths correspond to different algorithms:
\begin{enumerate}
    \item the top path (in blue) corresponds to a brute-force algorithm, 
    \item the ``middle path'' crossing upwards (red then pink then blue) corresponds to a compositional algorithm, but a \textit{slow} one and
    \item the bottom path (in red) corresponds to a fast (\(\fpt\)-time) dynamic programming algorithm.
\end{enumerate}
The commutativity of the squares (and of the diagram of the whole) represents correctness of the algorithm since it implies that it is equivalent (in terms of its returned solution) to the brute force approach; this is \textbf{Algorithmic Compositionality}. The ingredient upon which all of our reasoning hinges is the assumption that the arrow labeled ``\textit{GlobalSolSpace}'' is a sheaf. This is \textbf{Representational Compositionality} and it implies the commutativity of the whole diagram. 

Algorithmic Compositionality on its own is interesting, but it is only useful when it is paired with running time guarantees. To that end, precise statements about the running time of subroutines of the algorithm can already be gleaned from Diagram~\ref{diagram:ENGLISH}. Indeed note that the arrow labeled ``\textit{LocalSolSpace}'' (resp. ``\textit{is empty (locally)?}'') corresponds to computing local solution spaces (resp. answering the decision problem locally) in order to turn decompositions of data into decompositions of solution spaces (resp. Booleans). These arrows correspond to local computations done a linear\footnote{(linear in the number of component pieces of the decomposition)} number of times and thus one can intuit that the running time will be dominated by the arrow suggestively named ``Algorithm''. Most of our technical work will be devoted to proving the existence of this functor (Theorem~\ref{thm:A-is-functor}) and in showing that it is efficiently computable (Theorem~\ref{thm:filtering-algorithm}). This functor is not only \textit{a} sensible choice, but it is also in some sense the \textit{canonical} one: it is exactly what one comes to expect using standard category theoretic machinery\footnote{The functor arises via the \textit{monad} given by the \(\mathsf{colim} \dashv \mathsf{const}\) adjunction of diagrams}. All of these arguments will hinge on Representational Compositionality: they require `\textit{GlobalSolSpace}' to be a sheaf.

Given this context, we can now offer an informal statement of our main algorithmic contribution (Theorem~\ref{thm:filtering-algorithm}). It states that, given any problem encoded as a sheaf with respect to the topology given by the decompositions of data, there is an algorithm which solves its associated decision problem in time that grows 
\begin{enumerate}
    \item linearly in the number of constituent parts of the decompositions and 
    \item boundedly (often exponentially) in terms of the internal complexity of the constituent parts.
\end{enumerate}
Formally the theorem and the computational task it achieves read as follows. (Sections~\ref{sec:problems-as-functors} and Section~\ref{sec:compositional-data} are devoted to building up all of the necessary resources to make sense of this statement.) 

\begin{framed}
    \noindent \(\cat{C}-\textsc{SheafDecision}\) \\
    \noindent \textbf{Input:} a sheaf \(\mathcal{F} \colon \cat{C} \to \finset^{op}\) on the site \((\cat{C}, \decomp)\) where \(\cat{C}\) is a small adhesively cocomplete category, an object \(c \in \cat{C}\) and a \(C\)-valued structured decomposition \(d\) for \(c\).\\
    \noindent \textbf{Task:} letting \(\mathbf{2}\) be the two-object category \(\bot \to \top\) and letting \(\mathrm{dec} \colon \finset \to \mathbf{2}\) be the functor taking a set to \(\bot\) if it is empty and to \(\top\) otherwise, compute \( \mathrm{dec}^{op}\mathcal{F}c \).
\end{framed}

\begin{theorem}\label{thm:filtering-algorithm}
Let \(G\) be a finite, irreflexive, directed graph without antiparallel edges and at most one edge for each pair of vertices. Let \(\cat{D}\) be a small adhesively cocomplete category, let \(\Fa \colon \cat{D}^{op} \to \finset \) be a presheaf and let \(\cat{C}\) be one of \(\{\cat{D}, \:\cat{D}_{mono}\}\). If \(\Fa\) is a sheaf on the site \(\bigl(\cat{C}, \decomp|_{\cat{C}}\bigr)\) and if we are given an algorithm \(\Aa_{\Fa}\) which computes \(\Fa\) on any object \(c\) in time \(\alpha(c)\), then there is an algorithm which, given any \(\cat{C}\)-valued structured decomposition \(d \colon \int G \to \cat{C}\) of an object \(c \in \cat{C}\) and a feedback vertex set \(S\) of \(G\), computes \(\mathsf{dec} \Fa c \) in time \[\Oh(\max_{x \in VG}\alpha(dx) + \kappa^{|S|}\kappa^2)|EG|\] where \(\kappa = \max_{x \in VG}|\Fa dx|\). 
\end{theorem}

Theorem~\ref{thm:filtering-algorithm} enjoys three main strengths: 
\begin{enumerate}
    \item it allows us to recover some algorithmic results on graphs (for instance dynamic programming algorithms, for \textsc{\(H\)-Coloring} and for \textsc{\(H\)-ReflColoring}\footnote{Here, given a fixed graph \(H\), one is asked to determine whether a given graph \(G\) admits a reflexive homomorphism onto \(H\) (where a reflexive homomorphism is a vertex map \(f: VG \to VH\) such that, for all vertex pairs \((x,y)\) in \(G\), if \(f(x)f(y) \in EH\), then \(xy \in EG\).}) and
    \item it allows us to generalize both decompositions and dynamic programming thereupon to other kinds of structures (not just graphs) and 
    \item it is easily implementable. 
\end{enumerate}
We shall now briefly expand on the last two of these points. Notice that, since Theorem~\ref{thm:filtering-algorithm} applies to any adhesive category, we automatically obtain algorithms on a host of other structures encoded as any category of \(\cat{C}\)-sets~\cite{C-sets, LackAdhesive}. This is a remarkably large class of structures of which the following are but a few examples: \begin{enumerate*} \item databases,\item simple graphs, \item directed graphs, \item directed multigraphs, \item hypergraphs, \item directed hypergraphs, \item simplicial complexes, \item circular port graphs~\cite{libkind2021operadic} and \item half-edge graphs. \end{enumerate*}

We note that categorical perspective allows us to pair -- virtually effortlessly -- our meta-theoretic results (specifically the special case relating to tree-shaped decompositions) with a practical, runnable implementation~\cite{str_decomps_julia_implementation} in the AlgebraicJulia Ecosystem~\cite{evan_patterson_algebraicjuliacatlabjl_2020}. 

Finally, from a high-level view, our approach affords us another insight (already noted by Bodlaender and Fomin~\cite{bodlaender2002tree}): it is not the width of the decompositions of the inputs that matters; instead it is the \textit{width of the decompositions of the solutions spaces} that is key to the algorithmic bounds.

\paragraph{Outline} Our results rely on the following ingredients: \begin{enumerate*}
    \item encoding computational problems as functors
    \item describing structural compositionality via diagrams (and in particular a special class thereof suited for algorithmic manipulation called \textit{structured decompositions}),
    \item proving that these give rise to Grothendieck topologies and finally
    \item proving our main algorithmic result. 
\end{enumerate*}
In Section~\ref{sec:problems-as-functors} we explain how to view computational problems as functors. In Section~\ref{sec:diagrams-as-gr-topos} we provide background on structured decompositions and we prove that they yield Grothendieck topologies (Corollary~\ref{corollary:decomps-as-topologies}). We establish our algorithmic meta theorem (Theorem~\ref{thm:filtering-algorithm}) in Section~\ref{sec:deciding-sheaves} and we provide a discussion of our practical implementation~\cite{str_decomps_julia_implementation} of these results in AlgebraicJulia in Section~\ref{sec:implementation}. Finally we provide a discussion of our contributions and of opportunities for future work in Section~\ref{sec:discussion}.

\paragraph{Acknowledgements:} The authors would like to thank Evan Patterson for thought-provoking discussions which influenced the development of this paper and for providing key insights that improved the implementation of structured decompositions in AlgebraicJulia~\cite{evan_patterson_algebraicjuliacatlabjl_2020}. 

\subsection{Notation} Throughout the paper we shall assume familiarity with typical category theoretic notation as that found in Riehl's textbook~\cite{riehl2017category}. We use standard notation from sheaf theory and for any notion not explicitly defined here or in Appendix~\ref{appendix:sheaf-theory}, we refer the reader to Rosiak's textbook~\cite{rosiak-book}. Finally, for background on graph theory we follow Diestel's notation~\cite{Diestel2010GraphTheory} while we use standard terminology and notions from parameterized complexity theory as found for example in the textbook by Cygan et al.~\cite{cygan2015parameterized}.

\section{Computational Problems as Functors}\label{sec:problems-as-functors}
Computational problems are assignments of data -- thought of as \textit{solution spaces} -- to some class of input objects. We think of them as functors \(\Fa \colon \cat{C} \to \cat{Sol}\) taking objects of some category \(\cat{C}\) to objects of some appropriately chosen category \(\cat{Sol}\) of solution spaces. Typically, since solution spaces are prohibitively large, rather than computing the entire solution space, one instead settles for more approximate representations of the problem in the form of decision or optimization or enumeration problems etc.. For instance one can formulate \textit{an} \define{\(\Fa\)-decision problem} as a composite of the form \[\cat{C} \xrightarrow{\Fa} \cat{Sol} \xrightarrow{\decision} \mathbf{2}\] where \(\decision\) is a functor into \(\mathbf{2}\) (the walking arrow category) mapping solutions spaces to either \(\bot\) or \(\top\) depending on whether they witness yes- or no-instances to \(\Fa\). 

\paragraph{Some examples} familiar to graph theorists are \(\textsc{GaphColoring}_n\), \(\textsc{VertexCover}_k\) and Odd Cycle Transversal (denoted \(\textsc{OTC}_k\))
which can easily be encoded as contravariant functors into the category \(\finset\) of finite sets, as we shall now describe. 

\(\textsc{GraphColoring}_n\) is the easiest problem to define; it is simply the representable functor
\(\cat{SimpFinGr}(-, K_n) \colon \cat{SimpFinGr}^{op} \to \finset\) taking each graph \(G\) to the set of all homomorphisms from \(G\) to the n-vertex irreflexive complete graph. One turns this into decision problems by taking  \( \decision \colon \finset \to \mathbf{2} \) to be functor which takes any set to \(\bot\) if and only if it is empty.

For \(\textsc{VertexCover}_k\) and \(\textsc{OTC}_k\) we will instead work with the subcategory \[\cat{SimpFinGr}_{mono} \hookrightarrow \cat{SimpFinGr}\] of finite simple graphs and monomorphisms (subgraphs) between them. Notice that, if \(H'\) is a subgraph of a graph \(G'\) -- witnessed by the injection \(g \colon H' \hookrightarrow G'\) -- which satisfies some subgraph-closed property \(P\) of interest and if \(f \colon G \hookrightarrow G'\) is any monomorphism, then the pullback of \(g\) along \(f\) will yield a subgraph of \(G\) which also satisfies property \(P\) (since \(P\) is subgraph-closed). In particular this shows that, for any such subgraph-closed property \(P\) (including, for concreteness, our two examples of being either a vertex cover or an odd cycle transversal of size at most \(k\)) the following is a contravariant `functor by pullback' into the category of posets
\[F_P \colon \cat{SimpFinGr}_{mono} \to \cat{Pos}.\]
Letting \(U\) being the forgetful functor taking posets to their underlying sets, then we can state the corresponding decision problems (including, in particular \(\textsc{VertexCover}_k\) and \(\textsc{OCT}_k\)) as the following composite \[\cat{SimpFinGr}_{mono}^{op} \xrightarrow{F_P} \cat{Pos} \xrightarrow{U} \finset \xrightarrow{\decision} \mathbf{2}.\]

\section{Compositional Data \& Grothendieck Topologies}
\label{sec:compositional-data}
Parameterized complexity~\cite{cygan2015parameterized} is a two-dimensional framework for complexity theory whose main insight is that one should not analyze running times only in terms of the total input size, but also in terms of other \textit{parameters} of the inputs (such as measures of \textit{compositional structure}~\cite{cygan2015parameterized}). Here we represent compositional structure via \textit{diagrams}: we think of an object \(c \in \cat{C}\) obtained as the colimit of a diagram \(d \colon \cat{J} \to \cat{C}\) as being \textit{decomposed} by \(d\) into smaller constituent pieces. This section is split into two parts. In Section~\ref{sec:diagrams-as-gr-topos} we will show that diagrams yield Grothendieck topologies on adhesive categories. Then in Section~\ref{sec:structured-decomps} we focus on a special class of diagrams (suited for algorithmic manipulation) called \textit{structured decompositions}~\cite{structured-decompositions}. Roughly they consist of a collection of objects in a category and a collection of \textit{spans} which relate these objects (just like the edges in a graph relate its vertices).

\subsection{Diagrams as Grothendieck Topologies}\label{sec:diagrams-as-gr-topos}
In this section we will need \textit{adhesive categories}. We think of these as categories of data ``in which pushouts of monomorphisms exist and <<behave more or less as they do in the category of sets>>, or equivalently in any topos.''~\cite{nlab-adhesive-cat} Adhesive categories encompass many mathematical structures raging from topological spaces (or indeed any \textit{topos)} to familiar combinatorial structures such as: \begin{enumerate*} \item databases,\item simple graphs, \item directed graphs, \item directed multigraphs, \item hypergraphs, \item directed hypergraphs (i.e. Petri nets), \item simplicial complexes, \item circular port graphs~\cite{libkind2021operadic} and \item half-edge graphs \end{enumerate*}.

\begin{definition}[Adhesive category~\cite{LackAdhesive}]\label{def:adhesive}
A category \(\cat{C}\) is said to be \define{adhesive} if
\begin{enumerate}
    \item \(\cat{C}\) has pushouts along monomorphisms;
    \item \(\cat{C}\) has pullbacks;
    \item pushouts along monomorphisms are van Kampen.
\end{enumerate}
In turn, a pushout square of the form  
\[\begin{tikzcd}
	&& C \\
	A &&& B \\
	& D
	\arrow["g", from=2-1, to=3-2]
	\arrow["n"', from=2-4, to=3-2]
	\arrow["f"', from=1-3, to=2-4]
	\arrow["m"{pos=0.4}, from=1-3, to=2-1]
\end{tikzcd}\] 
is said to be \define{van Kampen} if, for any commutative cube as in the following diagram which has the above square as its bottom face, 
\[\begin{tikzcd}
	&& {C'} \\
	{A'} &&& {B'} \\
	& {D'} & C \\
	A &&& B \\
	& D
	\arrow["g", from=4-1, to=5-2]
	\arrow["n"', from=4-4, to=5-2]
	\arrow["f"', from=3-3, to=4-4]
	\arrow["m"{pos=0.4}, from=3-3, to=4-1]
	\arrow["a", from=2-1, to=4-1]
	\arrow["d", from=3-2, to=5-2]
	\arrow["c", from=1-3, to=3-3]
	\arrow["b", from=2-4, to=4-4]
	\arrow["{g'}", from=2-1, to=3-2]
	\arrow["{n'}"'{pos=0.6}, from=2-4, to=3-2]
	\arrow["{f'}", from=1-3, to=2-4]
	\arrow["{m'}"', from=1-3, to=2-1]
\end{tikzcd}\]
the following holds: if the back faces are pullbacks, then the front faces are pullbacks if and only if the top face is a pushout.
\end{definition}

We will concern ourselves with diagrams whose morphisms are monic; we call these \textit{submonic diagrams}.

\begin{definition}\label{def:monic-diagram}
    A \define{submonic diagram in \(\cat{C}\) of shape \(\cat{J}\)} is a diagram \(d: \cat{J} \to \cat{C}\) which preserves monomorphisms and whose domain is a finite category with all arrows monic. We say that a category \(\cat{C}\) is \define{adhesively cocomplete} if (1) it is adhesive and (2) every submonic diagram in \(\cat{C}\) has a colimit.
\end{definition}

The following result (\cite[Lemma 5.10]{structured-decompositions} restated here for convenience) states that we can associate -- in a functorial way -- to each object \(c\) of an adhesively cocomplete category \(\cat{C}\) the set of submonic diagrams whose colimit is \(c\).

\begin{lemma}[\cite{structured-decompositions}]\label{lemma:monic-diagrams-functorial-by-pointwise-pullback}
Let $\cat{C}$ be a small adhesively cocomplete category. For any arrow \(f \colon x \to y\) in \(\cat{C}\) and any diagram \(d_y \colon \cat{J} \to \cat{C}\) whose colimit is \(y\) we can obtain a diagram \(d_x \colon \cat{J} \to \cat{C}\) whose colimit is \(x\) by point-wise pullbacks of \(f\) and the arrows of the colimit cocone \(\Lambda\) over \(d_y\).
\[\begin{tikzcd}
	{X = \mathsf{colim} d_x} & {y = \mathsf{colim} d_y} \\
	{d_x J} & d_yJ
	\arrow["\Lambda"', Rightarrow, from=2-2, to=1-2]
	\arrow["f", from=1-1, to=1-2]
	\arrow[Rightarrow, from=2-1, to=2-2]
	\arrow[Rightarrow, from=2-1, to=1-1]
	\arrow["\lrcorner"{anchor=center, pos=0.125, rotate=90}, draw=none, from=2-1, to=1-2]
\end{tikzcd}\]
\end{lemma}

Note that we cannot do without requiring the diagrams to be monic in Lemma~\ref{lemma:monic-diagrams-functorial-by-pointwise-pullback} since, although adhesive categories have all pushouts of monic spans, they need not have pushouts of arbitrary spans (see Definition~\ref{def:adhesive}). 

As we shall now see, adhesively cocomplete categories have just enough structure for us to define a Grothendieck topologies where covers are given by colimits of monic subcategories~\footnote{For the reader concerned with size issues, observe that: (1) given categories \(\cat{C}\) and \(\cat{D}\), the functor category \([C,D]\) is small whenever \(\cat{C}\) and \(\cat{D}\) also are; and (2) since we are concerned with diagrams whose domains have finitely many objects and morphisms, one has that the collection of diagrams which yield a given object as a colimit is indeed a set.}. We will use this result to prove (Corollary~\ref{corollary:decomps-as-topologies}) that structured decompositions yield the desired Grothendieck topologies.

\begin{theorem}\label{thm:diagrams-as-topologies}
If \(\cat{C}\) is a small adhesively cocomplete category, then, denoting by \(\Lambda_d\) the colimit cocone of any diagram \(d\), the following is a functor by pullback. 
\begin{align*}
    \submon &\colon \cat{C}^{op} \to \mathbf{Set} \\
    \submon &\colon c \mapsto \{\Lambda_d \mid d  \text{ a submonic diagram and } \mathsf{colim} d = c\} \nonumber \\ &\quad \quad \quad \bigcup \{\{f\} \mid f \colon a \xrightarrow{\cong} c \text{ an iso }\}.
\end{align*}
Furthermore we have that: 
\begin{itemize}
    \item for any such \(\cat{C}\) the pair \((\cat{C}, \: \submon)\) is a site and
    \item denoting by \(\cat{C}_{mono}\) the subcategory of \(\cat{C}\) having the same objects of \(\cat{C}\), but only the monomorphisms of \(\cat{C}\), we have that the following pair is a site \[\bigl( \cat{C}_{mono}, \: \submon|_{\cat{C}_{mono}} \bigr).\]
\end{itemize}
\end{theorem}
\begin{proof}
The fact that \(\submon\) is a contravariant functor by pullback follows from the observation that pullbacks preserve isomorphisms together with Lemma~\ref{lemma:monic-diagrams-functorial-by-pointwise-pullback}.

Now notice that it suffices to show that \(\submon\) defines a Grothendieck pre-topology (Definition~\ref{def:pre-topology}) since it is known that every Grothendieck pre-topology determines a genuine Grothendieck topology~\cite{rosiak-book}. We do this only for the first case (the second case is proved analogously) and to that end we proceed by direct verification of the axioms (consult Definition~\ref{def:pre-topology} in Appendix~\ref{appendix:sheaf-theory}). First of all note that Axiom~\ref{axiom:pretopology-1} holds trivially by the definition of \(\submon\). It is also immediate that Axiom~\ref{axiom:pretopology-2} holds since we have just shown in the first part of this theorem that \(\submon\) is a contravariant ``functor by pullback''. Axiom~\ref{axiom:pretopology-3} is also easily established as follows: 
\begin{itemize}
    \item since a diagram of diagrams is a diagram, we have that a colimit of colimits is a colimit and thus the resulting colimit cocone is in the cover; and
    \item if we are given any singleton cover consisting of an isomorphism \(\{f \colon b \xrightarrow{\cong} c\}\), then there are two cases: \begin{enumerate*}
    \item we are given another such cover \(f' \colon a \xrightarrow{\cong} b\) then the composite \(ff'\) is an isomorphism into \(c\) and thus is in the coverage and 
    \item if the cover on \(b\) is a colimit cocone \(\Lambda_d\) of submonic diagram \(d\), then the composition of the components of the cocone (which are monomorphism) with isomorphism \(f\) determines a diagram of subobjects of \(c\) and this diagram will have \(c\) as its colimit, as desired.
    \end{enumerate*}
\end{itemize}
\end{proof}

Recall from Section~\ref{sec:introduction} that some computational problems (e.g. \(\textsc{VertexCover}_k\) and \(\textsc{OCT}_k\)) are stated on the category \(\cat{SimpFinGr}_{mono}\) of finite simple graphs and the monomorphisms between them. Indeed note that, if we were to attempt to state these notions on the category \(\cat{SimpFinGr}\), then we would lose the ability to control the cardinalities of the vertex covers (resp. odd cycle transversals, etc.) since, although the pullback of a vertex cover along an epimorphism is still a vertex cover (resp. odd cycle transversals, etc.), the size of the resulting vertex cover obtained by pullback may increase arbitrarily. Thus, if we want to think of computational problems (the specific ones mentioned above, but also more generally) as sheaves, it is particularly helpful to be able to use structured decompositions as Grothendieck topologies on `categories of monos' (i.e. \(\cat{C}_{mono}\) for some adhesive \(\cat{C}\)) which will fail to have pushouts in general. 

\paragraph{Tangential Remark} Note that similar ideas to those of Theorem~\ref{thm:diagrams-as-topologies} and Corollary~\ref{corollary:decomps-as-topologies} can be used to define sheaves on `synthetic spaces' (i.e. mixed-dimensional manifolds encoded as structured decompositions of manifolds). Here the idea is that one would like to define sheaves on topological spaces obtained by gluing manifolds of varying dimensions together while simultaneously remembering that the resulting object should be treated as a `synthetic mixed-dimensional manifold' (i.e. structured decompositions of manifolds). The fact that one can use structured decompositions as topologies implies that one can speak of sheaves which \textit{are aware} of the fact that the composite topological space needs to be regarded as a mixed-dimensional manifold. This is beyond the scope of the present paper, but it is a fascinating direction for further work.

\subsection{Structured Decompositions}\label{sec:structured-decomps}
For our algorithmic purposes, it will be convenient to use \textit{structured decompositions}, rather than arbitrary diagrams, to define Grothendieck topologies. This is for two reasons: (1) structured decompositions are a class of particularly nice diagrams which are easier to manipulate algorithmically when the purpose is to construct colimits via recursive algorithms and (2) although it is true that, given its colimit cone, one can always turn a diagram into a diagram of spans (as we argue in Corollary~\ref{corollary:decomps-as-topologies}), this operation is computationally expensive (especially when, as we will address later, one is dealing with decompositions of prohibitively large objects such as \textit{solutions spaces}). 

\paragraph{Background on Structured Decompositions} \textit{Structured decompositions}~\cite{structured-decompositions} are category-theoretic generalizations of many combinatorial invariants -- including, for example, tree-width~\cite{Bertele1972NonserialProgramming, halin1976s, RobertsonII}, layered tree-width~\cite{layered-tw-dujmovic2017, layered-tw-shahrokhi2015}, co-tree-width~\cite{co-tw-Sousa}, $\mathcal{H}$-tree-width~\cite{H-tw-conference} and graph decomposition width~\cite{CARMESIN2022101} -- which have played a central role in graph theory and parameterized complexity. 

\begin{figure}[h!]
    \centering\includegraphics[width=0.8\textwidth]{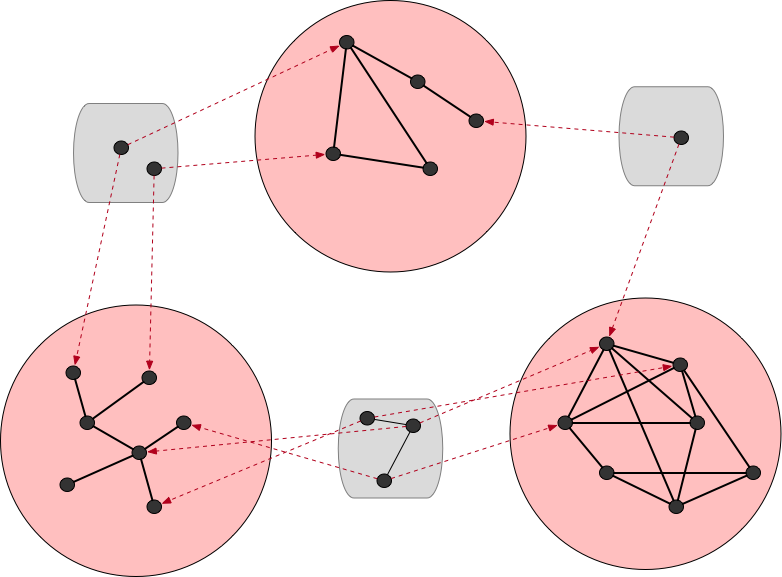}
    \caption{A \(\GRhomo\)-valued structured decomposition of shape \(K_3\). The spans (which in the figure above are monic) are drawn (componentwise) dotted in red. The bags are highlighted in pink and the adhesions are highlighted in gray.}
    \label{fig:str-decomp-example}
\end{figure}

Intuitively structured decompositions should be though of as generalized graphs: whereas a graph is a \textit{relation} on the \textit{elements} of a \textit{set}, a structured decomposition is a \textit{generalized relation} (a collection of \textit{spans}) on a collection of \textit{objects} of a \textit{category}. For instance consult Figure~\ref{fig:str-decomp-example} for a drawing of a graph-valued structured decomposition of shape \(K_3\) (the complete graph on three vertices).

Formally, for any category \(\cat{C}\), one defines \(\cat{C}\)-valued structured decompositions as certain kinds of diagrams in \(\cat{C}\) as in the following definition. For the purposes of this paper, we will only need the special case of structured decompositions given by diagrams whose arrows are all monic; however, for completeness, we note that the theory of structured decompositions does not rely on such a restriction. Thus, to ease legibility and since we will only work with monic structured decompositions in the present document, we will drop the adjective `monic' and instead speak of \textit{structured decompositions} (or simply `decompositions').

\begin{definition}[Monic Structured Decomposition]\label{def:str-decomps-obs}
Given any finite graph \(G\) viewed as a functor \(G \colon \cat{GrSch} \to \set\) where \(\cat{GrSch}\) is the two object category generated by
\[
\begin{tikzcd}
	E & V
	\arrow["s", shift left=1, from=1-1, to=1-2]
	\arrow["t"', shift right=1, from=1-1, to=1-2]
\end{tikzcd}
\] one can construct a category \(\int G\) with an object for each vertex of \(G\) and an object from each edge of \(G\) and a span joining each edge of \(G\) to each of its source and target vertices. This construction is an instance of the more general notion of Grothendieck construction. Now, fixing a base category $\cat{K}$ we define a \define{$\cat{K}$-valued structured decomposition of shape $G$} (see Figure~\ref{fig:str-decomp-example}) as a diagram of the form \[d \colon \int G \to \cat{K}\] whose arrows are all \textit{monic} in \(\cat{K}\). Given any vertex $v$ in $G$, we call the object $d v$ in $\cat{K}$ the \define{bag of $d$ indexed by $v$}. Given any edge $e = xy$ of $G$, we call \(de\) the \define{adhesion} indexed by \(e\) and we call the span $dx \leftarrow {d e} \rightarrow dy$ the \define{adhesion-span indexed by $e$}.
\end{definition}

\begin{definition}[Morphisms of Structured Decompositions]\label{def:str-decomps-homs}
A \define{morphism of $\cat{K}$-valued structured decompositions from $d_1$ to $d_2$} is a pair \[(F, \eta) \colon \bigl(\int G_1 \xrightarrow{d_1} \cat{K}\bigr) \to \bigl(\int G_2 \xrightarrow{d_2} \cat{K}\bigr)\] as in the following diagram where $F$ is a functor $F: \int G_1 \to \int G_2$ and $\eta$ is a natural transformation $\eta : d_1 \Rightarrow d_2 F$ as in the following diagram.
\[\begin{tikzcd}
	{\int G_1} && {\int G_2} \\
	&& {} \\
	& {\cat{K}}
	\arrow[""{name=0, anchor=center, inner sep=0}, "{d_1}"', from=1-1, to=3-2]
	\arrow["{d_2}", from=1-3, to=3-2]
	\arrow["F", from=1-1, to=1-3]
	\arrow["\eta", shorten <=10pt, shorten >=10pt, Rightarrow, from=0, to=1-3]
\end{tikzcd}\]
\end{definition}

The point of such abstraction is that one can now speak in a unified way about decompositions of many different kinds of objects. Indeed this is precisely our approach in this paper: we relate the compositional structure of the \textit{inputs} of a computational problem to a corresponding compositional structure of the \textit{solutions spaces} of the problem. To do so, we leverage the fact that the construction of the category of structured decompositions~\cite[Prop. 3.3]{structured-decompositions} is functorial~\cite[Corol. 3.4]{structured-decompositions}. In particular there is a functor \[\sd \colon \cat{Cat} \to \cat{Cat}\] taking any category \(\cat{C}\) to the category \(\sd \cat{C}\) of \(\cat{C}\)-valued structured decompositions\footnote{In order to maintain notational consistency with the prequel~\cite{structured-decompositions} and in order to remind the reader that we are only working throughout with \textit{monic} structured decompositions, we will keep the subscript \(m\) in the notation of the category \(\sd \cat{C}\).}; this is summarized in Definition~\ref{def:cat-str-decomps} below.

\begin{definition}[The category of Structured Decompositions]\label{def:cat-str-decomps}
\textbf{Category of Decompositions.} Fixing a category $\cat{K}$, $\cat{K}$-valued structured decompositions (of any shape) and the morphisms between them (as in Definition~\ref{def:str-decomps-obs}) form a category $\sd(\cat{K})$ called the \emph{category of $\cat{K}$-valued structured decompositions}. Furthermore, this construction is functorial: there is a functor $\sd : \CAT \to \CAT$ which takes any category $\cat{K}$ to the category $\sd(\cat{K})$ and every functor $\Phi: \cat{K} \to \cat{K}'$ to the functor $\sd(\Phi): \sd(\cat{K}) \to \sd(\cat{K}')$ defined on objects as 
\[\sd(\Phi) \colon \bigl(\int G \xrightarrow{d} \cat{K} \bigr ) \mapsto \bigl( \sd(\Phi) d \colon  \int G \xrightarrow{d} \cat{K} \xrightarrow{\Phi} \cat{K}' \bigr ) \] 
and on arrows as 
\[
\adjustbox{scale=0.8}{
\begin{tikzcd}
	{\int G_1} && {\int G_2} && {\int G_1} && {\int G_2} \\
	&& {} & {} & {} &&& {} \\
	& {\cat{K}} &&&& {\cat{K}'} & {.}
	\arrow[""{name=0, anchor=center, inner sep=0}, "{\sd_\Phi d_1}"', from=1-5, to=3-6]
	\arrow["{\sd(\Phi) d_2}", from=1-7, to=3-6]
	\arrow["F", from=1-5, to=1-7]
	\arrow[""{name=1, anchor=center, inner sep=0}, "{d_1}"', from=1-1, to=3-2]
	\arrow["{d_2}", from=1-3, to=3-2]
	\arrow["F", from=1-1, to=1-3]
	\arrow["{\sd(\Phi) }", maps to, from=2-3, to=2-5]
	\arrow["\Phi\eta", shorten <=13pt, shorten >=13pt, Rightarrow, from=0, to=1-7]
	\arrow["\eta", shorten <=8pt, Rightarrow, from=1, to=1-3]
\end{tikzcd}
}
\]
\end{definition}

The fact that \(\sd\) is a functor means that every functor \(F \colon \cat{C} \to \cat{Sol}\) which encodes some computational problem lifts to a functor \[\sd_F \colon \sd \cat{C} \to \sd \cat{Sol}\] taking structured decompositions of the \textit{inputs to a computational problem} (i.e. decompositions valued in \(\cat{C}\)) to those valued in the \textit{solution spaces of the problem} (i.e. decompositions valued in \(\cat{Sol}\)) by evaluating \(F\) on the bags and adhesions of the objects of \(\sd \cat{C}\). We will make use of this fact in Section~\ref{sec:deciding-sheaves} to easily encode local computations on the bags and adhesions of decompositions and for the implementations of our results (see Section~\ref{sec:implementation}).

\paragraph{Structured Decompositions as Grothendieck Topologies}

We are finally ready to show that structured decompositions yield Grothendieck topologies on adhesive categories. We will proceed in much the same way as we did for arbitrary submonic diagrams. Indeed observe~\cite[Corol. 5.11]{structured-decompositions} that, for any small adhesively cocomplete category \(\cat{C}\), one can define (again by point-wise pullbacks as in Theorem~\ref{thm:diagrams-as-topologies}) a subfunctor \(\decomp\) of the functor \(\submon\) (defined in Theorem~\ref{thm:diagrams-as-topologies}) as follows (where we once again denote by \(\Lambda_d\) the set of arrows in the colimit cocone over \(d\)).
\begin{align}
        \decomp &\colon \cat{C}^{op} \to \mathbf{Set} \label{def:decomp_functor}\\
        \decomp &\colon c \mapsto \{\Lambda_d \mid \colim d = c \text{ and } d \in \sd\cat{C}\} \bigcup \{\{f\} \mid f \colon a \xrightarrow{\cong} c \text{ an iso }\}. \nonumber
\end{align}

\begin{corollary}\label{corollary:decomps-as-topologies}
The pairs \((\cat{C}, \decomp)\) and  \(\bigl( \cat{C}_{mono}, \: \decomp|_{\cat{C}_{mono}} \bigr)\) are sites where \(\cat{C}\) is any small adhesively cocomplete category  and \(\decomp\) is the functor of Equation~\eqref{def:decomp_functor}.
\end{corollary}
\begin{proof}
Every \(\cat{C}\)-valued structured decomposition induces a diagram in \(\cat{C}\). Conversely, since \(\cat{C}\) is adhesive, we have that colimit cocones of monic subcategories will consist of monic arrows. Thus, by taking pairwise pullbacks of the colimit arrows, one can recover a structured decomposition with the same colimit (because \(\cat{C}\) is adhesive) and having all adhesions monic.
\end{proof}

In the general case, there are going to be many distinct Grothendieck topologies that we could attach to a particular category, yielding different sites. In building a particular site, we have in mind the sheaves that will be defined with respect to it. It would be nice if we could use some relationship between sites to induce a relationship between the associated sheaves, and reduce our work by lifting the computation of sheaves with respect to one Grothendieck topology from the sheaves on another (appropriately related) Grothendieck topology. As it turns out, the collection of Grothendieck topologies on a category $\cat{C}$ are partially ordered (by inclusion), and established results show us how to exploit this to push sheaves on one topology to sheaves on others. \par 
We might first define maps between arbitrary coverings as follows. 
\begin{definition}
\normalfont 
A \define{morphism of coverings} from $\mathcal{V} = \{V_j \rightarrow V\}_{j \in J}$ to $\mathcal{U} = \{U_i \rightarrow U\}_{i \in I}$ is an arrow $V \rightarrow U$ together with a pair $(\sigma, f)$ comprised of 
\begin{itemize}
	\item a function $\sigma: J \rightarrow I$ on the index sets, and 
	\item a collection of morphisms $f = \{f_j: V_j \rightarrow U_{\sigma(j)}\}_{j \in J}$ with 
		\begin{center}
			\begin{tikzcd}
				V_j \arrow[d] \arrow[r] & U_{\sigma(j)} \arrow[d] \\ 
				V \arrow[r] & U
			\end{tikzcd}
		\end{center}
		commuting for all $j \in J$.
\end{itemize}
\end{definition}
With this notion of morphisms, we could define a category of the coverings on $\cat{C}$. However, we will mostly be concerned with a particular case of such morphisms of coverings, namely where $V = U$ and $V \rightarrow U$ is just the identity map. In this case, we say that $\mathcal{V}$ is a \define{refinement} of $\mathcal{U}$.
\begin{definition}
	\normalfont 
	Let $\cat{C}$ be a category and $\mathcal{U} = \{U_i \rightarrow U\}_{i \in I}$ be a family of arrows. Then a \define{refinement} $\mathcal{V} = \{V_j \rightarrow U\}_{j \in J}$ is a family of arrows such that for each index $j \in J$ there exists some $i \in I$ such that $V_j \rightarrow U$ factors through $U_i \rightarrow U$.
\end{definition}
In sieve terms, we can simplify the above to say that, given $\mathcal{U} = \{U_i \rightarrow U\}$ and $\mathcal{V} = \{V_j \rightarrow U\}$, $\mathcal{V}$ is a refinement of $\mathcal{U}$ if and only if the sieve generated by $\mathcal{V}$ is contained in (a sub-sieve of) the sieve generated by $\mathcal{U}$. \par 
Moreover, observe that any covering is a refinement of itself, and a refinement of a refinement is a refinement. As such, refinement gives us an ordering on the coverings of an object $U$. We can use this to define the following relation for Grothendieck topologies. 
\begin{definition}
	\normalfont 
	 Let $\cat{C}$ be a category and $J, J'$ two Grothendieck topologies on $\cat{C}$. We say that $J$ is \define{subordinate to} $J'$, denoted $J \leqslant J'$, provided every covering in $J$ has a refinement that is a covering in $J'$. \par 
	 If $J \leqslant J'$ and $J' \leqslant J$, then $J$ and $J'$ are \define{equivalent}.
\end{definition}
This really says not just that any covering in $J$ is a covering in $J'$---which makes us say that the topology $J'$ is finer than $J$---but that for any covering in $J$, there exists a refinement among the coverings in $J'$. \par 
In sieve terms, $J \leqslant J'$ just says that for any object $U$ and any sieve $S$ considered a covering sieve by $J$, there exists a sieve $S'$ that is considered a covering sieve by $J'$, where all arrows of $S$ are also in $S'$. In particular, two topologies will be equivalent if and only if they have the same sieves. \par 
Now we come to the point of these definitions. Using the main sheaf definition---i.e., a functor $F$ is a sheaf with respect to a topology $J$ if and only if for any sieve $S$ belonging to $J$ the induced map\footnote{Here we write \(\mathbf{y}_U\) to denote the Yoneda embedding at \(U\).} 
\begin{equation*}
F(U) \simeq \text{Hom}(\mathbf{y}_U, F) \rightarrow \text{Hom}(S, F)
\end{equation*}
is a bijection---together with the previous definition (expressed in terms of sieves), we get the following result.
\begin{proposition}\label{prop:sheaves-on-subordinate-topologies}	Let $J, J'$ be two Grothendieck topologies on a category $\cat{C}$. If $J$ is  subordinate to $J'$, i.e., $J \leqslant J'$, then every presheaf that satisfies the sheaf condition for $J'$ also satisfies it for $J$.
\end{proposition}
In other words: if $J$ is subordinate to $J'$, then we know that any sheaf for $J'$ will automatically be a sheaf for $J$. 
\begin{proof} 
We could show the result by deploying the notion of morphisms of sites (see Mac Lane and Moerdijk~\cite{maclane2012sheaves}) $f: (\cat{C}, J) \rightarrow (\cat{D}, K)$, as a certain ``cover-preserving" morphism, where this is got by setting $\cat{D} = \cat{C}$ and $K = J'$, and using the identity functor $\cat{C} \rightarrow \cat{C}$, taking any object to itself and any $J$-covering to itself as well. Precomposition with $f$ gives a functor between the category of presheaves (going the other way), and we then use the pushforward (or direct image) functor $f_*: \cat{Sh}(\cat{C}, J') \rightarrow \cat{Sh}(\cat{C}, J)$, where this is the restriction of the precomposition with $f$ down to sheaves, to push a sheaf with respect to $J'$ to a sheaf $f_* F$ on $J$. Details can be found in the literature.
\end{proof}
As a particular case, two equivalent topologies have the same sheaves. \par 
The main take-away here is, of course, that if we can show that a particular presheaf is in fact a sheaf with respect to a topology $J'$, and if $J$ is another topology which we can establish is subordinate to $J'$, then we will get for free a sheaf with respect to $J$ (and with respect to any other subordinate topology, for that matter). Moreover, in particular, if there is a \textit{finest} cover, verifying the sheaf axiom there will guarantee it for all covers. \par 
Let's now connect these established results to the Grothendieck topologies that we defined in Theorem~\ref{thm:diagrams-as-topologies} and Corollary~\ref{corollary:decomps-as-topologies}, namely $\submon$ and $\decomp$, where the latter is a subfunctor of the former. We can focus on showing how $\submon$ relates to another Grothendieck topology of interest.
\begin{proposition}\label{prop:subobject-topology}
Fix any adhesively cocomplete category \(\cat{C}\). The topology given by $\submon$ on \(\cat{C}_{mono}\) is subordinate to the subobject topology.
\end{proposition}
\begin{proof}
The subobject topology just takes for coverings of an object $c$ (equivalence classes of) monomorphisms whose union is $c$. Need to show that any covering in the monic diagram topology can be refined by a covering in the subobject topology. This follows immediately from the definition of $\submon$ in Corollary~\ref{corollary:decomps-as-topologies}.
\end{proof}
This, combined with Proposition~\ref{prop:sheaves-on-subordinate-topologies}, yields the following result which can be particularly convenient when proving that a given presheaf is indeed a sheaf with respect to the decomposition topology.
\begin{proposition}\label{prop:subordinate-topology} 
\normalfont 
Fix any adhesively cocomplete category \(\cat{C}\) and any presheaf \(\Fa: \cat{C}^{op} \to \set\). If \(\Fa\) is a sheaf on \(\cat(C)_{mono}\) with respect to the subobject topology, then it will automatically give us a sheaf on the site \(\bigl( \cat{C}_{mono}, \: \decomp|_{\cat{C}_{mono}} \bigr)\).
\end{proposition}

\section{Deciding Sheaves}\label{sec:deciding-sheaves}
Adopting an algorithmically-minded perspective, the whole point of Thm~\ref{thm:diagrams-as-topologies} and Corollary~\ref{corollary:decomps-as-topologies} is to speak about computational problems which can be solved via compositional algorithms. To see this, suppose we are given a sheaf \(\mathcal{F} \colon \cat{C}^{op} \to \finset \) with respect to the site \( (\cat{C}, \decomp) \). We think of this sheaf as representing a \textit{computational problem}: it specifies which solutions spaces are associated to which inputs. We've already seen a concrete example of such a construction in Section~\ref{sec:introduction}; namely the coloring functor \[\cat{SimpFinGr}(-, K_n) \colon \cat{SimpFinGr}^{op} \to \finset\] which takes each graph to the set of all \(n\)-colorings of that graph. 


In this paper we focus on decision problems and specifically on the sheaf decision problem (defined in Section~\ref{sec:introduction}) which we recall below for ease of reference.

\begin{framed}
    \noindent \(\cat{C}-\textsc{SheafDecision}\) \\
    \noindent \textbf{Input:} a sheaf \(\mathcal{F} \colon \cat{C} \to \finset^{op}\) on the site \((\cat{C}, \decomp)\) where \(\cat{C}\) is a small adhesively cocomplete category, an object \(c \in \cat{C}\) and a \(C\)-valued structured decomposition \(d\) for \(c\).\\
    \noindent \textbf{Task:} letting \(\mathbf{2}\) be the two-object category \(\bot \to \top\) and letting \(\mathrm{dec} \colon \finset \to \mathbf{2}\) be the functor taking a set to \(\bot\) if it is empty and to \(\top\) otherwise, compute \( \mathrm{dec}^{op}\mathcal{F}c \).
\end{framed}

In this section we will show that the sheaf decision problem lies in \(\fpt\) under the dual parameterization of the width of the given structured decomposition and the feedback vertex number\footnote{A \define{feedback vertex set} in a graph \(G\) is a vertex subset \(S \subseteq V(G)\) of \(G\) whose removal from \(G\) yields an acyclic graph. The \define{feedback vertex number} of a graph \(G\) is the minimum size of a feedback vertex set in \(G\).} of the shape of the decomposition. Notice that, when the decomposition is tree-shaped, we recover parameterizations by our abstract analogue of tree-width (which, note, can be instantiated in \textit{any} adhesive category, not just that of graphs). 

To that end, notice that, if \(\cat{C}\) has colimits, then, since sheaves on this site send colimits of decompositions to limits of sets~\cite{rosiak-book}, the following diagram will always commute (see Bumpus, Kocsis \& Master~\cite{structured-decompositions} for a reference). 
\begin{equation}\label{diagram:compositional-square-when-C-has-colimits}
\begin{tikzcd}
	{\cat{C}} && {\finset^{op}} && {\mathbf{2}^{op}} \\
	{\sd\cat{C}} && {\sd\finset^{op}}
	\arrow["{\mathsf{colim}}", color={rgb,255:red,92;green,92;blue,214}, from=2-1, to=1-1]
	\arrow["{\mathsf{colim}}"', color={rgb,255:red,214;green,92;blue,92}, from=2-3, to=1-3]
	\arrow["{\mathcal{F}}"{description}, color={rgb,255:red,92;green,92;blue,214}, from=1-1, to=1-3]
	\arrow["{\sd\mathcal{F}}"{description}, color={rgb,255:red,214;green,92;blue,92}, from=2-1, to=2-3]
	\arrow["{comm.}"{description}, draw=none, from=2-1, to=1-3]
	\arrow["{\mathrm{dec}^{op}}"{description}, from=1-3, to=1-5]
\end{tikzcd}
\end{equation}
For concreteness, let us unpack what this means. The blue path corresponds to first gluing the constituent parts of the decomposition together (i.e. forgetting the compositional structure) and then solving the problem on the entire output. On the other hand the red path corresponds to a compositional solution: one first evaluates \(\Fa\) on the constituent parts of the decomposition and then joins\footnote{Note that the colimit functor \(\colim \colon \sd \finset ^{op} \to \finset^{op}\) -- which is in red in Diagram~\ref{diagram:compositional-square-when-C-has-colimits} -- is a \textit{limit of sets}; we invite the reader to keep this in mind throughout.} these solutions together to find a solution on the whole. Thus, since this diagram commutes for any sheaf \(\Fa\), we have just observed that there is a compositional algorithm for \textsc{SheafDecision}. 

Unfortunately the approach we just described is still very inefficient since for any input \(c\) and no matter which path we take in the diagram, we always end-up computing all of \(\Fa(c)\) which is very large in general (for example, the images of the coloring sheaf which we described above have cardinalities that grow exponentially in the size of the input graphs). One might hope to overcome this difficulty by lifting \(\mathsf{dec}\) to a functor from \(\finset^{op}\)-valued structured decompositions to \(\mathbf{2}^{op}\)-valued structured decompositions as is shown in the following diagram. 
\[\begin{tikzcd}
	{\cat{C}} && {\finset^{op}} && {\mathbf{2}^{op}} \\
	{\sd\cat{C}} && {\sd\finset^{op}} && {\sd\mathbf{2}^{op}}
	\arrow["{\mathsf{colim}}", color={rgb,255:red,92;green,92;blue,214}, from=2-1, to=1-1]
	\arrow["{\mathsf{colim}}"', color={rgb,255:red,214;green,92;blue,92}, from=2-3, to=1-3]
	\arrow["{\mathcal{F}}"{description}, color={rgb,255:red,92;green,92;blue,214}, from=1-1, to=1-3]
	\arrow["{\sd\mathcal{F}}"{description}, color={rgb,255:red,214;green,92;blue,92}, from=2-1, to=2-3]
	\arrow["{comm.}"{description}, draw=none, from=2-1, to=1-3]
	\arrow["{\mathrm{dec}^{op}}"{description}, from=1-3, to=1-5]
	\arrow["{\sd\mathrm{dec}^{op}}"{description}, from=2-3, to=2-5]
	\arrow["{\mathsf{colim} = \land}"', from=2-5, to=1-5]
\end{tikzcd}\]
\begin{figure}
    \centering
    \includegraphics[width=0.7\textwidth]{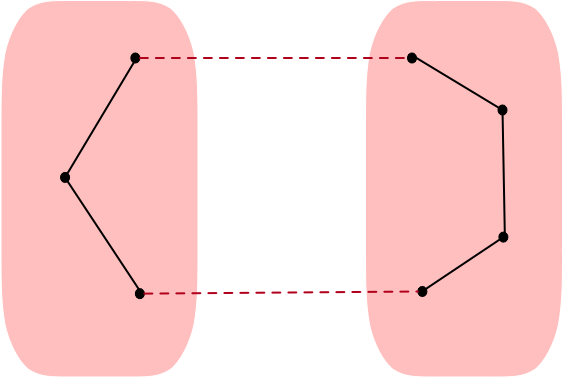}
    \caption{A structured decomposition \(P_3 \leftarrow \overline{K_2} \to P_4\) which decomposes a \(5\)-cycle into two paths on \(2\) and \(3\) edges respectively}
    \label{fig:5-cycle-decomp-example}
\end{figure}
However, this too is to no avail: the right-hand square of the above diagram does \textit{not} commute in general. To see why, consider the \(2\)-coloring sheaf \(\cat{SimpFinGr}(-, K_2)\) and let \(d\) be the structured decomposition \(P_3 \leftarrow \overline{K_2} \to P_4\) which decomposes a \(5\)-cycle into two paths on \(2\) and \(3\) edges respectively (see Figure~\ref{fig:5-cycle-decomp-example}). The image of \(d\) under \(\sd\mathsf{dec}^{op} \circ \sd \Fa\) is \(\top \leftarrow \top \to \top\) (since the graphs \(P_3\), \(\overline{K_2}\) and \(P_4\) are all \(2\)-colorable) and the colimit of this diagram in \(\mathbf{2}^{op}\) (i.e. a limit in \(\mathbf{2}\); i.e. a conjunction) is \(\top\). However, \(5\)-cycles are not \(2\)-colorable.

Although this counterexample might seem discouraging, we will see that the approach we just sketched is very close to the correct idea. Indeed in the rest of this section we will show (Theorem~\ref{thm:A-is-functor}) that there is an endofunctor \[\mathcal{A} \colon \sd \finset^{op} \to \sd \finset^{op}\] 
which makes the following diagram commute. 
\begin{equation}\label{diagram:algorithmic-goal}
\begin{tikzcd}
	{\cat{C}} && {\finset^{op}} && {\mathbf{2}^{op}} \\
	{\sd\cat{C}} && {\sd\finset^{op}} & {\sd\finset^{op}} & {\sd\mathbf{2}^{op}}
	\arrow["{\mathsf{colim}}", color={rgb,255:red,92;green,92;blue,214}, from=2-1, to=1-1]
	\arrow[""{name=0, anchor=center, inner sep=0}, "{\mathsf{colim}}"', color={rgb,255:red,214;green,92;blue,92}, from=2-3, to=1-3]
	\arrow["{\mathcal{F}}"{description}, color={rgb,255:red,92;green,92;blue,214}, from=1-1, to=1-3]
	\arrow["{\sd\mathcal{F}}"{description}, color={rgb,255:red,214;green,92;blue,92}, from=2-1, to=2-3]
	\arrow["{comm.}"{description}, draw=none, from=2-1, to=1-3]
	\arrow["{\mathrm{dec}^{op}}"{description}, from=1-3, to=1-5]
	\arrow["{\mathcal{A}}", from=2-3, to=2-4]
	\arrow["{\sd\mathrm{dec}^{op}}", from=2-4, to=2-5]
	\arrow[""{name=1, anchor=center, inner sep=0}, "\land"', from=2-5, to=1-5]
	\arrow["{comm.\text{ (by Thm.~\ref{thm:A-is-functor})}}"{description, pos=0.6}, draw=none, from=0, to=1]
\end{tikzcd}
\end{equation}
Moreover we will show that this functor \(\mathcal{A}\) can be computed efficiently (this is Theorem~\ref{thm:filtering-algorithm}, which we mentioned in the introduction).

\paragraph{When \(\cat{C}\) fails to have colimits.} Before we move on to defining the object and morphism maps of \(\Aa\) (which -- as we shall prove in Theorem~\ref{thm:A-is-functor} -- assemble into a functor), we will briefly note how the entire discussion mentioned above can be applied even to sheaves defined on the site \(\bigl( \cat{C}_{mono}, \decomp|_{\cat{C}_{mono}}\bigr)\). Recall from Sections~\ref{sec:problems-as-functors} and~\ref{sec:compositional-data} that there are many situations in which we would like to speak of computational problems defined on the category \(\cat{C}_{mono}\) having the same objects of \(\cat{C}\), but only those morphisms in \(\cat{C}\) which are monos. In this case, since \(\cat{C}_{mono}\) will fail to have colimits in general, we trivially cannot deduce the commutativity of the `square of compositional algorithms' (Diagram~\ref{diagram:compositional-square-when-C-has-colimits}) since one no longer has the colimit functor (namely \(\colim\) which is marked in blue in Diagram~\ref{diagram:compositional-square-when-C-has-colimits}). However, note that in either case -- i.e. taking any \(\cat{K} \in \{\cat{C}, \cat{C}_{mono}\} \) -- if \(\Fa\) is a sheaf with respect to the site \((\cat{K}, \decomp)\), then we always have that, for any covering \[(\kappa \in \cat{K}, \: d \in \decomp \kappa)\] we always have that 
\begin{equation}\label{equation:sheafy-commutativity}
    (\mathsf{dec}^{op} \circ \Fa)(\kappa) = (\mathsf{dec}^{op} \circ \const \circ \sd\Fa)(d).
\end{equation}
To state this diagrammatically, we can invoke the Grothendieck construction\footnote{Recall that the Grothendieck construction produces from any functor \(F \colon \cat{A} \to \cat{B}\) a \textit{category} denoted \(\int F\). For details on the construction and its properties, we refer the reader to Riehl's textbook~\cite{riehl2017category}}. For convenience we spell-out the definition of the category \(\int \decomp\); it is defined as having
\begin{itemize}
    \item \textbf{objects} given by pairs \((c , d)\) with \(c \in \cat{C}\) and \(d \in \decomp(c) \)
    \item \textbf{arrows} from \((c , d)\) to \((c' , d')\) are given by morphisms in \(\cat{C}\) of the form \( f_0 \colon c \to c' \) such that \(\decomp(f)(d') = d\).
\end{itemize}
Notice that this allows us to restate Equation~\eqref{equation:sheafy-commutativity} as the commutativity of the following diagram. 
\[\begin{tikzcd}
	& {\cat{C}} && {\finset^{op}} && {\mathbf{2}^{op}} \\
	\int\decomp \\
	& {\sd\cat{C}} && {\sd\finset^{op}}
	\arrow["{\mathsf{colim}}"', color={rgb,255:red,214;green,92;blue,92}, from=3-4, to=1-4]
	\arrow["{\mathcal{F}}"{description}, color={rgb,255:red,92;green,92;blue,214}, from=1-2, to=1-4]
	\arrow["{\sd\mathcal{F}}"{description}, color={rgb,255:red,214;green,92;blue,92}, from=3-2, to=3-4]
	\arrow["{comm.}"{description}, draw=none, from=3-2, to=1-4]
	\arrow["{\mathrm{dec}^{op}}"{description}, from=1-4, to=1-6]
	\arrow["{\mathsf{fst}}", from=2-1, to=1-2]
	\arrow["{\mathsf{snd}}"', from=2-1, to=3-2]
\end{tikzcd}\]
where \(\mathsf{fst}\) and \(\mathsf{snd}\) are the evident projection functors, whose object maps are respectively \(\mathsf{fst} \colon (c,d) \mapsto c\) and \(\mathsf{snd} \colon (c,d) \mapsto d\). The point of all of this machinery is that it allows us to restate our algorithmic goal (namely  Diagram~\eqref{diagram:algorithmic-goal} -- whose commutativity we shall prove in Theorem~\ref{thm:A-is-functor}) in the following manner when our category \(\cat{C}\) of inputs does not have colimits. (Notice that the following diagram is the formal version of Diagram~\ref{diagram:ENGLISH} which we used to sketch our `roadmap' in Section~\ref{sec:introduction}.)

\begin{equation}\label{diagram:algorithmic-goal-NO-COLIMITS}
\begin{tikzcd}
	& {\cat{C}} && {\finset^{op}} && {\mathbf{2}^{op}} \\
	{\int \decomp} \\
	& {\sd\cat{C}} && {\sd\finset^{op}} & {\sd\finset^{op}} & {\sd\mathbf{2}^{op}}
	\arrow[""{name=0, anchor=center, inner sep=0}, "{\mathsf{colim}}"', color={rgb,255:red,214;green,92;blue,92}, from=3-4, to=1-4]
	\arrow["{\mathcal{F}}", color={rgb,255:red,92;green,92;blue,214}, from=1-2, to=1-4]
	\arrow["{\sd\mathcal{F}}"', color={rgb,255:red,214;green,92;blue,92}, from=3-2, to=3-4]
	\arrow["{comm.}", shift right=5, draw=none, from=3-2, to=1-4]
	\arrow["{\mathrm{dec}^{op}}", from=1-4, to=1-6]
	\arrow["{\mathcal{A}}"', from=3-4, to=3-5]
	\arrow["{\sd\mathrm{dec}^{op}}"', from=3-5, to=3-6]
	\arrow[""{name=1, anchor=center, inner sep=0}, "\land"', from=3-6, to=1-6]
	\arrow["{\mathsf{fst}}", from=2-1, to=1-2]
	\arrow["{\mathsf{snd}}"', from=2-1, to=3-2]
	\arrow["{comm. \text{ (by Thm.~\ref{thm:A-is-functor})}}"{pos=0.7}, shift right=3, draw=none, from=0, to=1]
\end{tikzcd}
\end{equation}

All that remains to be done is to establish that the functor \(\Aa\) in Diagram~\ref{diagram:algorithmic-goal-NO-COLIMITS}. As we shall see, we will define this functor by making use of the monad \(\Ta\) defined as follows. Recall that, for any category \(\cat{K}\) with pullbacks and colimits, there is the following adjunction~\cite{structured-decompositions} (which is a special case of the analogous adjunction in categories of diagrams)
\[ \begin{tikzcd} \sd \cat{K} \ar[r,bend left,"\cat{colim}",pos=.4] \ar[r,phantom,"\bot"]& \ar[l,bend left,"\cat{const}",pos=.6] \cat{K}\end{tikzcd}\]
where the functor \(\const\) takes objects of \(\cat{K}\) to trivial decompositions (having a single bag) and the functor \(\colim\) takes decompositions to their colimits. Now, taking \(\cat{K} = \finset^{op}\), we shall denote by \(\Ta\) the monad 
\begin{align}
    \Ta &\colon \sd\finset^{op} \to \sd\finset^{op} \label{eqn:monad}\\
    \Ta &\colon  d \mapsto (\const \circ \colim)(d) \nonumber
\end{align}
given by this adjunction. In the proof of the following theorem we will make use of this monad to finally establish the existence of the desired functor \(\Aa\). 

\begin{theorem}\label{thm:A-is-functor}
There is a functor \(\mathcal{A} \colon \sd \finset^{op} \to \sd \finset^{op}\) such that:
\begin{enumerate}[label=\textbf{(A\arabic*)}]
    \item there is a natural isomorphism \label{thm:A-is-functor-point-1}
\[\begin{tikzcd}
	{\finset^{op}} \\
	\\
	{\sd\finset^{op}} && {\sd\finset^{op}}
	\arrow["{\mathcal{A}}"', from=3-1, to=3-3]
	\arrow[""{name=0, anchor=center, inner sep=0}, "{\colim}"', from=3-3, to=1-1]
	\arrow["{\colim}", from=3-1, to=1-1]
	\arrow["\cong", draw=none, from=3-1, to=0]
\end{tikzcd}\]
    \item 
    there are natural transformations \(\alpha^1\) and \(\alpha^2\) which factor the unit of  \(\eta\) of the monad \(\mathcal{T}\) -- see Functor~\eqref{eqn:monad} -- as in the following diagram \label{thm:A-is-functor-point-2}
    
\[\begin{tikzcd}
	{\id_{\sd\finset^{op}}} && {\mathcal{T}} \\
	& {\mathcal{A}}
	\arrow["\eta", Rightarrow, from=1-1, to=1-3]
	\arrow["{\alpha^1}"', Rightarrow, from=1-1, to=2-2]
	\arrow["{\alpha^2}"', Rightarrow, from=2-2, to=1-3]
\end{tikzcd}\]
    \item and the following diagram (which is Diagram~\eqref{diagram:algorithmic-goal-NO-COLIMITS}, restated) commutes \label{thm:A-is-functor-point-3} \[\begin{adjustbox}{scale=0.9}\end{adjustbox}\]
\end{enumerate}
\end{theorem}
One should think of the functor \(\mathcal{A}\) of Theorem~\ref{thm:A-is-functor} as a pre-processing routine by which one filters-out those local solutions which cannot be extended to global solutions. Furthermore, notice that there is a sense in which \(\Aa\) is \textit{canonical} since it arises from the monad \(\Ta\). Indeed as we shall see, although most of the proof is focused on the derivation of \(\Aa\), none of the steps in this derivation are ad-hoc: they are the `obvious steps' which are completely determined by the properties of the monad \(\Ta\). We find this to be a great benefit of our categorical approach. 
\subsection{Proof of Theorem~\ref{thm:A-is-functor}} 
The proof is structured as follows: we shall begin by making use of the monad \(\Ta\) to define the functor \(\Aa \colon \sd \finset^{op} \to \sd \finset^{op}\); then we will point out how various facts accumulated through the derivation of \(\Aa\) easily imply the points~\ref{thm:A-is-functor-point-1},~\ref{thm:A-is-functor-point-2} and~\ref{thm:A-is-functor-point-3} listed in the statement of the theorem. 
\subsubsection{Gathering Intuition}
Notice that, if we momentarily disregard any consideration of efficiency of computation, then it is easy to find a candidate for the desired functor \(\mathcal{A}\). To see this, consider what happens when we pass a \(\finset^{op}\)-valued structured decomposition \(d\) through the monad \(\Ta\). We will have that \(\Ta d\) consists of a structured decomposition of shape \(K_1\) (the one-vertex complete graph) whose bag consists of the solution space we seek to evaluate (namely the colimit of \(d\) in \(\finset^{op}\)). Thus we will clearly have that Diagram~\eqref{diagram:algorithmic-goal} commutes if we replace \(\mathcal{A}\) by \(\mathcal{T}\).

Although \(\mathcal{T}\) is not efficiently computable (since, as we mentioned earlier, the set \(\colim d\) might be very large) we will see that the unit \(\eta\) of the monad \(\mathcal{T}\) given by the adjunction \(\const \vdash \mathsf{colim}\) is the crucial ingredient needed to specify (up to isomorphism) the desired functor \(\mathcal{A}\) (shown to be efficiently computable later in this section.) Thus, towards defining \(\Aa\), we will first spell-out the definition of \(\eta\). 

Notice that \(\eta\) yields a morphism of structured decompositions \[\eta_d : d \to \mathcal{T}d\] which by definition consists of a pair \((\eta_d^0, \eta_d^1)\) where \(\eta_d^0 \colon \int G \to \int K_1\) is a functor and \(\eta_d^1 \colon d \Rightarrow  \mathcal{T}(d)\eta_d^0\) is a natural transformation as in the following diagram.
\begin{equation}\label{diagram:unit-of-monad}
    \begin{tikzcd}
	{\int G} && {\int K_1} \\
	\\
	& {\finset^{op}}
	\arrow[""{name=0, anchor=center, inner sep=0}, "{\mathcal{T}d}", from=1-3, to=3-2]
	\arrow[""{name=1, anchor=center, inner sep=0}, "d"', from=1-1, to=3-2]
	\arrow["{\eta_d^0}", from=1-1, to=1-3]
	\arrow["{\eta_d^1}", shorten <=9pt, shorten >=9pt, Rightarrow, from=1, to=0]
\end{tikzcd}
\end{equation}
Now notice that the diagram above views \(d\) and \(\Ta d\) as \(\finset^{op}\)-valued structured decompositions. Instead, in the interest of convenience and legibility, we will "slide the op" in Diagram~\eqref{diagram:unit-of-monad} so as to think of \(\eta_d\) as a morphism of \(\finset\)-valued structured \define{co-decompositions} (i.e. contravariant structured decompositions); in other words we will rewrite Diagram~\eqref{diagram:unit-of-monad} as follows. 
\begin{equation}\label{diagram:unit-of-monad}
    \begin{tikzcd}
	{(\int G)^{op}} && {(\int K_1)^{op}} \\
	\\
	& {\finset}
	\arrow[""{name=0, anchor=center, inner sep=0}, "{\mathcal{T}d}", from=1-3, to=3-2]
	\arrow[""{name=1, anchor=center, inner sep=0}, "d"', from=1-1, to=3-2]
	\arrow["{(\eta_d^0)^{op}}", from=1-1, to=1-3]
	\arrow["{\lambda^d}", shorten <=9pt, shorten >=9pt, Rightarrow, from=0, to=1]
\end{tikzcd}
\end{equation}
For clarity, we will take a moment to spell out what this change of perspective entails. First of all observe that \(\finset\)-valued co-decompositions associate to each edge of \(G\) a \textit{cospan} of sets whose legs are epimorphisms (since they are monomorphisms in \(\finset^{op}\)). The natural transformation \(\lambda^d\) is simply the transformation \(\eta_d^1\) when its components are viewed not as morphisms in \(\finset^{op}\), but as morphisms of \(\finset\) (we choose to denote it \(\lambda^d\) simply to avoid notational confusion). The components of \(\lambda^d\) are functions (projections) from the single bag of \(\Ta d\) to each one the bags of \(d\) as shown below.
\begin{equation}\label{eqn:def-limit-cone-leg}
    \lambda^{d}_{x} \colon \mathcal{T}d \to dx.
\end{equation}
In particular these are given by the \textit{limit cone}\footnote{Once again we remind the reader that taking the \textbf{colimit} in \(\finset^{op}\) of a \(\finset^{op}\)-valued decomposition \(d \colon \int G \to \finset^{op}\) will yield the same result as first viewing \(d\) as a \(\finset\)-valued \textbf{co}-decomposition \(d \colon (\int G)^{op} \to \finset\) and evaluating its \textbf{limit} in \(\finset\).} sitting above \(d\) with apex \(\Ta d\). 

\subsubsection{Defining the object map of \(\Aa\).} With these observations in mind, we can now define the object map
\begin{equation}\label{eqn:A-obj-cpt}
     \mathcal{A}_0 \colon \Bigl( \bigl(\int G)^{op} \xrightarrow{d} \finset \Bigr) \mapsto \Bigl(  \bigl(\int G\bigr)^{op} \xrightarrow{\mathcal{A}_0 d} \finset \Bigr)
\end{equation}
of the functor \(\mathcal{A}\) which we are seeking to establish. It maps \(d\) to the decomposition obtained by restricting all of the bags and adhesions of \(d\) to the images of the legs of the limit cone with apex \(\Ta d\) (we encourage the reader to consult Figure~\ref{fig:A-object-map} for a visual aid). 
Spelling this out formally, the map \(\mathcal{A}_0\) takes a structured co-decomposition \(d \colon (\int G)^{op} \to \finset\) to the structured co-decomposition \(\mathcal{A}_0 d \colon (\int G)^{op} \to \finset\) which has the same shape as \(d\) and which is defined as follows: 
\begin{align}
    \mathcal{A}_0 d &\colon (x \in \int G) \mapsto \image \lambda^{d}_{x} \label{eqn:A-ob-cpt-details}\\
    \mathcal{A}_0 d &\colon (e \xleftarrow{f_{x,e}} x \in \morphisms \bigl (\int G \bigr)^{op}) \mapsto \bigl( \image \lambda^{d}_{e} \xleftarrow{f_{x,e}|_{\image \lambda^{d}_{x}}} \image \lambda^{d}_{x} \bigr ). \nonumber 
\end{align}

\begin{figure}
    \begin{adjustbox}{scale=0.85}
    \begin{tikzcd}
	{G \cong K_2} & {d \colon (\int G)^{op} \to \finset} && {\mathcal{T}d} && {\mathcal{A}_0d \colon (\int G)^{op} \to \finset} \\
	x & dx &&&& {\image \lambda_x^d} \\
	& de && {\lim d} && {\image(f_{x,e} \circ \lambda_x^{d})} \\
	y & dy &&&& {\image \lambda_y^d}
	\arrow["e", no head, from=2-1, to=4-1]
	\arrow["{f_{x,e}}", color={rgb,255:red,214;green,92;blue,92}, from=2-2, to=3-2]
	\arrow["{f_{y,e}}"', color={rgb,255:red,214;green,92;blue,92}, from=4-2, to=3-2]
	\arrow["{f_{x,e}|_{\image \lambda_{dx}}}", from=2-6, to=3-6]
	\arrow["{f_{y,e}|_{\image \lambda_{dy}}}"', from=4-6, to=3-6]
	\arrow["{\lambda^d_{y}}", from=3-4, to=4-2]
	\arrow["{\lambda^d_{x}}"', from=3-4, to=2-2]
    \end{tikzcd}
    \end{adjustbox}
    \caption{A visualisation of the definition (Equations~\eqref{eqn:A-obj-cpt} and~\eqref{eqn:A-ob-cpt-details}) of the object map of \(\mathcal{A}\) on a structured decomposition \(d\) of shape \(\vec{K}_2\) (the directed edge). The unit of the monad given by \(\mathcal{T}\) yields a morphism \((\eta_d^0, \eta_d^1) : d \to \mathcal{T}d\). The components of \(\eta_d^1\) (namely \(\eta_d^1(x)\) and \(\eta_d^1(y)\)) are morphisms in \(\finset^{op}\); when viewed as morphisms in \(\finset\), we denote them as \(\lambda^d_x\) and \(\lambda^d_y\). These morphisms are given by the limit cone over \(d\) (viewed as a diagram in \(\finset\)).}
    \label{fig:A-object-map}
\end{figure}
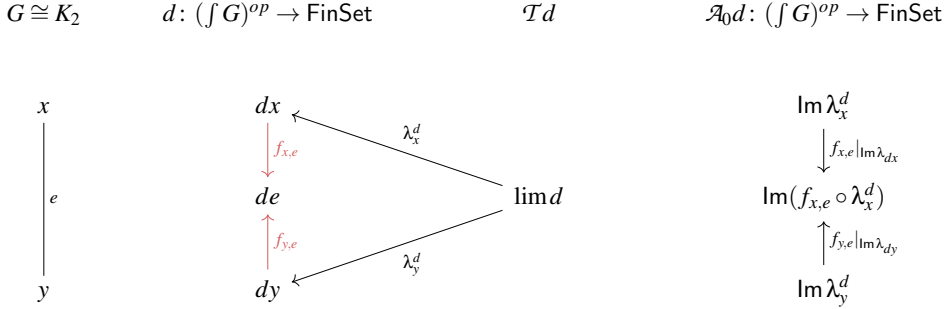

\subsubsection{Preliminaries of the definition of the morphism map of \(\Aa\).}
Now, towards defining the morphism map \(\mathcal{A}_1\), observe that, by the definition of \(\mathcal{A}_0\) via the unit \(\eta\) of \(\mathcal{T}\), we have that \(\eta\) factors as
\begin{equation}\label{eqn:def-nat-transfs-alpha-pt-1}
\begin{tikzcd}
	d &&& {\mathcal{A}_0d} &&& {\mathcal{T}d}
	\arrow["{\alpha^2_{\mathcal{A}d} = (\eta^0_d, \mathfrak{a}^2_d)}"', color={rgb,255:red,92;green,92;blue,214}, from=1-4, to=1-7]
	\arrow["{\eta_d = (\eta_d^0, \lambda^d)}", curve={height=-24pt}, from=1-1, to=1-7]
	\arrow["{\alpha^1_d = (\id_{\domain(d)}, \mathfrak{a}^1_d)}"', color={rgb,255:red,214;green,92;blue,92}, from=1-1, to=1-4]
\end{tikzcd}
\end{equation}
where, for any \(x \in \int G\), the maps \(\mathfrak{a}^1_d\) and \(\mathfrak{a}^2_{\mathcal{A}d}\) shown below 
\[\begin{tikzcd}
	{(\int G)^{op}} && {(\int G)^{op}} && {(\int K_1)^{op}} \\
	\\
	&& \finset
	\arrow[""{name=0, anchor=center, inner sep=0}, "d"'{pos=0.2}, from=1-1, to=3-3]
	\arrow["{\eta_d^0}", color={rgb,255:red,92;green,92;blue,214}, from=1-3, to=1-5]
	\arrow[""{name=1, anchor=center, inner sep=0}, "{\mathcal{A}d}"'{pos=0.3}, from=1-3, to=3-3]
	\arrow[""{name=2, anchor=center, inner sep=0}, "{\mathcal{T}d}"{pos=0.2}, from=1-5, to=3-3]
	\arrow["{\mathsf{id}}", color={rgb,255:red,214;green,92;blue,92}, from=1-1, to=1-3]
	\arrow["{\mathfrak{a}^2}", color={rgb,255:red,92;green,92;blue,214}, shorten <=7pt, shorten >=7pt, Rightarrow, from=2, to=1]
	\arrow["{\mathfrak{a}^1}", color={rgb,255:red,214;green,92;blue,92}, shorten <=7pt, shorten >=7pt, Rightarrow, from=1, to=0]
\end{tikzcd}\]
are respectively defined by restriction and corestriction\footnote{Dually to the notation used for restrictions, for any function \(f \colon A 
\to B\), we denote the \textbf{corestriction} of \(f\) to its image as \(f|^{\image f}\).} of \(\lambda^d\). Spelling this out, these maps are defined as the following morphisms in \(\finset\)
\begin{equation}\label{eqn:def-nat-transfs-alpha-pt-2}
\begin{tikzcd}
	{\mathfrak{a}_d^1(x) :=} & {\mathcal{A}dx} && dx \\
	{\mathfrak{a}_{\mathcal{A}d}^2 (x) :=} & {\mathcal{T}dK_1} && {\mathcal{A}dx}
	\arrow["{\id_{dx}|_{\image \lambda^{d}_{x}}}", hook, from=1-2, to=1-4]
	\arrow["{\lambda_x^d|^{{\image \lambda^{d}_{x}}}}", from=2-2, to=2-4]
\end{tikzcd}
\end{equation}
\subsubsection{Defining the morphism map of \(\Aa\).} Given any morphism \(q = (q^0, q^1)\) of \(\finset\)-valued structured \textit{co}-decompositions as in the following diagram
\begin{equation}\label{diagram:Q-arrow}
\begin{tikzcd}
	{(\int G)^{op}} && {(\int H)^{op}} \\
	\\
	& \finset
	\arrow["{q^0}", color={rgb,255:red,92;green,92;blue,214}, from=1-1, to=1-3]
	\arrow[""{name=0, anchor=center, inner sep=0}, "d"', from=1-1, to=3-2]
	\arrow[""{name=1, anchor=center, inner sep=0}, "{d'}", from=1-3, to=3-2]
	\arrow["{q^1}", color={rgb,255:red,92;green,92;blue,214}, shorten <=8pt, shorten >=8pt, Rightarrow, from=1, to=0]
\end{tikzcd}
\end{equation}
we have, by Diagram~\eqref{eqn:def-nat-transfs-alpha-pt-1} and the functoriality of the monad \(\Ta\) and the naturality of its unit, that the following diagram commutes.
\begin{equation}\label{diagram:pieces-needed-for-spelling-out-the-definition-of-morphism-cpt-of-A}
\begin{tikzcd}
	d &&& {\mathcal{A}_0d} &&& {\mathcal{T}d} \\
	\\
	{d'} &&& {\mathcal{A}_0d'} &&& {\mathcal{T}d'}
	\arrow["{\alpha^2_{\mathcal{A}d}}"', color={rgb,255:red,214;green,92;blue,92}, from=1-4, to=1-7]
	\arrow["{\eta_d = (\eta_d^0, \lambda^d)}", curve={height=-24pt}, from=1-1, to=1-7]
	\arrow["{\alpha^1_d}"', from=1-1, to=1-4]
	\arrow["q", from=1-1, to=3-1]
	\arrow["{\mathcal{T}q}", color={rgb,255:red,92;green,92;blue,214}, from=1-7, to=3-7]
	\arrow["{\alpha^1_{d'}}", from=3-1, to=3-4]
	\arrow["{\alpha^2_{\mathcal{A}d'}}", from=3-4, to=3-7]
	\arrow["{\eta_{d'}}"', curve={height=24pt}, from=3-1, to=3-7]
\end{tikzcd}
\end{equation}
Recall that each of the morphisms in the above diagram are themselves morphisms in a category of diagrams consisting of a functor and a natural transformation where the natural transformation points ``backwards'' (for example as in Diagram~\eqref{diagram:Q-arrow}). For instance recall that the morphism \(\eta_{d} \colon d \to \Ta d\) (drawn in Diagram~\eqref{diagram:unit-of-monad}) corresponds to the limit cone over \(d\).

We are seeking to prove the existence of a morphism \(\Aa_0 d \to \Aa_0 d'\) which commutes with Diagram~\eqref{diagram:pieces-needed-for-spelling-out-the-definition-of-morphism-cpt-of-A}. This will follow from the commutativity of Diagram~\eqref{diagram:pieces-needed-for-spelling-out-the-definition-of-morphism-cpt-of-A}. However, it deserves to be unpacked as follows.

Consider Diagram~\eqref{diagram:pieces-needed-for-spelling-out-the-definition-of-morphism-cpt-of-A}. The morphism \(\Ta q\) represents a function from the limit of \(d'\) to the limit of \(d\). The morphisms \(\eta_d\) and \(\eta_{d'}\) represent the limit cones sitting above \(d\) and \(d'\) respectively. The morphism \(q\) (defined in Diagram~\eqref{diagram:Q-arrow}) corresponds to assigning to each object \(x \in \domain d\) a morphism \[q^1_{x} \colon (q^0 x \in \domain d') \to (x \in \domain d).\] The commutativity of Diagram~\eqref{diagram:pieces-needed-for-spelling-out-the-definition-of-morphism-cpt-of-A} amounts to stating the following: for each \(x \in \domain d\) and each component \(q^1_{x}\) of \(q_1\), the range of \(q^1_{x}\) is completely contained in the images of the leg \(\lambda^d_x \colon \Ta d \to dx\) of the limit cone \(\lambda^d\) at \(x\). But then, this means that, by restricting \(q^1_{x}\) to the image of \(\lambda^{d'}x\), we obtain the desired morphism shown below (where recall \(\alpha^1_{d'} := (\mathsf{id}_{(\int G)^{op}}, \mathfrak{a}^1_{d'})\).
\[\begin{tikzcd}
	{(\int G)^{op}} && {(\int G)^{op}} \\
	\\
	& \finset
	\arrow[""{name=0, anchor=center, inner sep=0}, "{\mathcal{A}_0 d'}"', from=1-1, to=3-2]
	\arrow[""{name=1, anchor=center, inner sep=0}, "{\mathcal{A}_0 d}", from=1-3, to=3-2]
	\arrow["{\mathsf{id}}", from=1-1, to=1-3]
	\arrow["{\mathfrak{a}^1_{d'} \circ q^1}", shorten <=8pt, shorten >=8pt, Rightarrow, from=1, to=0]
\end{tikzcd}\]

In summary, this derivation allowed us to define the desired morphism \(\Aa q \colon \mathcal{A}_0d \to \mathcal{A}_0d'\) given by
\begin{equation}\label{eqn:A-morphism-map}
\mathcal{A}_1 \colon \bigl( (q_0, q_1) \colon d \to d' \bigr) \mapsto \bigl( (q_0, \mathfrak{a}^1_{d'} \circ q^1) \bigr).
\end{equation}
\subsubsection{Completing the proof}
The functoriality of \(\mathcal{A}\) is immediate: it clearly preserves identities and, since \(\Ta\) is a functor, one can see that \(\Aa\) preserves composition by inspection of Diagram~\eqref{diagram:pieces-needed-for-spelling-out-the-definition-of-morphism-cpt-of-A}. On the other hand Point~\ref{thm:A-is-functor-point-1} can be seen to follow by the definition of \(\mathcal{A}_0\) and the properties of limits in \(\finset\) (namely that, if \(X_1 \xleftarrow{\pi_1} X_1 \times_S X_2 \xrightarrow{\pi_2} X_2\) is the pullback of a span \(X_1 \xrightarrow{f_1} S \xleftarrow{f_2} X_2\), then  \(X_1 \times_S X_2\) will also be isomorphic to the pullback object of the span \(\image \pi_1 \xrightarrow{f_1|_{\image \pi_1}} S \xleftarrow{f_2|_{\image \pi_1}} \image \pi_2\)). Since the naturality of \(\alpha^1\) and \(\alpha^2\) is evident, Point~\ref{thm:A-is-functor-point-2} just amounts to the commutativity of Diagram~\eqref{eqn:def-nat-transfs-alpha-pt-1}. Finally, to show Point~\ref{thm:A-is-functor-point-3}, observe that, by the definition of the object map \(\mathcal{A}_0\) and since the only set with a morphism to the empty set is the empty set, we have that \(\land \circ \sd\mathsf{dec}^{op} \circ \mathcal{A} = \mathsf{dec}^{op} \circ \colim.\) 

\subsection{An Algorithm for Computing \(\Aa\)}
Recall that one should think of the functor \(\mathcal{A}\) of Theorem~\ref{thm:A-is-functor} as a pre-processing routine which filters-out those local solutions which cannot be extended to global solutions. This intuition suggests the following local filtering algorithm  according to which one filters pairs of bags locally by taking pullbacks along adhesions and then retaining only those local solutions which are in the image of the projection maps of the pullbacks. 
\begin{algorithm}
    \caption{Filtering}\label{algorithm:single-edge-filtering}
    \begin{algorithmic}
        \STATE \textbf{Input:} a \(\finset\)-valued structured co-decomposition \(d: (\int G)^{op} \to \finset\) and an edge \(e = xy\) in \(G\).
        \STATE -- compute the pullback \(dx \xleftarrow{\pi_x} dx \times_{de} dy \xrightarrow{\pi_y} dy\) of the cospan \(dx \xrightarrow{f_x} de_{x,y} \xleftarrow{f_y} dy\) associated to each edge \(e = xy\) in \(G\)
        \STATE -- let \(d_e\) be the decomposition obtained by replacing the cospan associated to \(e\) in \(d\) by the the cospan \[\image \pi_x \xrightarrow{f_x|_{\image \pi_x}} de_{x,y} \xleftarrow{f_y|_{\image \pi_y}} \image \pi_y,\]
        \RETURN \(d_e\).
    \end{algorithmic}
\end{algorithm}

Recursive applications of Algorithm~\ref{algorithm:single-edge-filtering} allow us to obtain the following algorithm for sheaf decision on \textit{tree}-shaped structured decompositions.
\begin{lemma}\label{lemma:tree-algorithm-correctness}
There is an algorithm (namely Algorithm~\ref{algorithm:filtering}) which correctly computes \(\mathcal{A}d\) on any input \(\finset\)-valued structured co-decomposition \(d \colon (\int G)^{op} \to \finset\) in time \(\Oh(\kappa^2)|EG|\) where \(\kappa = \max_{x \in VG} |dx|\) whenever \(G\) is a finite, irreflexive, directed tree.
\end{lemma}
\begin{proof}
Fixing any enumeration \(e_1, \dots, e_m\) of the edges of \(G\), we will show that that following recursive procedure (Algorithm~\ref{algorithm:filtering}) satisfies the requirements of the lemma when called on inputs \((d, (e_1, \dots, e_m))\). 
\begin{algorithm}
    \caption{Recursive filtering}\label{algorithm:filtering}
    \begin{algorithmic}
        \STATE \textbf{Input:} a \(\finset^{op}\)-valued structured decomposition \(d: (\int G)^{op} \to \finset\) and a list \(\ell\) of edges of \(G\).
        \IF{\(\ell\) is empty}
            \RETURN \(d\)
        \ELSE
            \STATE -- split \(\ell\) into its head-edge \(e\) and its tail \(\ell'\) 
            \STATE -- let \(d_{e}\) be the output of Algorithm~\ref{algorithm:single-edge-filtering} on inputs \((d, e)\)
            \STATE -- recursively call Algorithm~\ref{algorithm:filtering} on inputs \((d_e, \ell')\).  
        \ENDIF
    \end{algorithmic}
\end{algorithm}
Notice that in Algorithm~\ref{algorithm:single-edge-filtering} we always have an injection \(\image \pi_x \hookrightarrow dx\). By this fact and since Algorithm~\ref{algorithm:filtering} amounts to computing \(|EG|\) pullbacks in \(\finset\) and hence the running time bound is evident. Now suppose the input decomposition is tree-shaped (i.e. suppose that \(G\) is a tree); we will proceed by induction to show that the output of Algorithm~\ref{algorithm:single-edge-filtering} is isomorphic to \(\mathcal{A}d\). If \(G \cong K_1\) (the one-vertex complete graph), then the algorithm terminates immediately (since \(G\) has no edges) and returns \(d\). This establishes the base-case of the induction since \(Ad \cong d\) whenever \(G \cong K_1\). Now suppose \(|EG| > 0\) and let \(e = x_1x_2\) be the last edge of the tree \(G\) over which Algorithm~\ref{algorithm:filtering} iterates upon. The removal of \(e\) splits \(G\) into two sub-trees \(T_1 \hookrightarrow G\) and \(T_2 \hookrightarrow G\) containing the nodes \(x_1\) and \(x_2\) respectively. In turn these two trees induce two sub-decompositions \(\iota_i \colon d_i \hookrightarrow d \) of \(d\); these are given by the following composite functors: \[d_1 \colon (\int T_1)^{op} \hookrightarrow (\int G)^{op} \to \finset \quad \text{ and } \quad d_2 \colon (\int T_2)^{op} \hookrightarrow (\int G)^{op} \to \finset.\] 
By the induction hypothesis we have that, for each \(i \in \{1,2\}\), running Algorithm~\ref{algorithm:single-edge-filtering} on \(d_i\) yields an object \(\delta_i\) isomorphic to \(\mathcal{A} d_i \colon (\int T_i)^{op} \to \finset\). Thus, by Point~\ref{thm:A-is-functor-point-2} of Theorem~\ref{thm:A-is-functor}, the span of structured decompositions \(d_1 \xhookrightarrow{\iota_1} d \xhookleftarrow{\iota_2} d_2\) yields the following commutative diagram.
\begin{equation}\label{diagram:tree-alg-lemma-1}
\begin{tikzcd}
	{\mathcal{T}d_1} &&& {\mathcal{A}d_1} &&& {d_1} \\
	\\
	{\mathcal{T}d} &&& {\mathcal{A}d} &&& d \\
	\\
	{\mathcal{T}d_2} &&& {\mathcal{A}d_1} &&& {d_2}
	\arrow["{\alpha^1_{d_1}}", from=1-7, to=1-4]
	\arrow["{\alpha^2_{\mathcal{A}d_1}}", from=1-4, to=1-1]
	\arrow["{\alpha^1_{d_2}}", from=5-7, to=5-4]
	\arrow["{\alpha^2_{\mathcal{A}d_2}}", from=5-4, to=5-1]
	\arrow["{\eta_{d_1}}"{pos=0.7}, curve={height=30pt}, from=1-7, to=1-1]
	\arrow["{\eta_{d_2}}"{pos=0.7}, shift right=1, curve={height=30pt}, from=5-7, to=5-1]
	\arrow["{\iota_1}"{pos=0.3}, from=1-7, to=3-7]
	\arrow["{\iota_1}"'{pos=0.3}, from=5-7, to=3-7]
	\arrow["{\alpha^1_{d}}", from=3-7, to=3-4]
	\arrow["{\alpha^2_{\mathcal{A}d}}", from=3-4, to=3-1]
	\arrow["{\eta_d}"{pos=0.7}, curve={height=30pt}, from=3-7, to=3-1]
	\arrow["{\mathcal{A}_{\iota_1}}"{pos=0.3}, from=1-4, to=3-4]
	\arrow["{\mathcal{A}_{\iota_2}}"'{pos=0.3}, from=5-4, to=3-4]
	\arrow["{\mathcal{T}_{\iota_1}}"{pos=0.3}, from=1-1, to=3-1]
	\arrow["{\mathcal{T}_{\iota_2}}"'{pos=0.3}, from=5-1, to=3-1]
\end{tikzcd}
\end{equation}
Now consider the cospan \(dx_1 \xrightarrow{e_{x_1}} de \xleftarrow{e_{x_1}} dx_2\). Passing it through the functor \(\const\) to obtain morphisms \[\const e_{x_1} \colon \const d' \to d_i\] such that the following commutes.  
\[\begin{tikzcd}
	{\mathcal{T}d_1} &&& {\mathcal{A}d_1} &&& {d_1} \\
	\\
	{\mathcal{T}d} &&& {\mathcal{A}d} &&& d && {\const de} \\
	\\
	{\mathcal{T}d_2} &&& {\mathcal{A}d_1} &&& {d_2}
	\arrow["{\alpha^1_{d_1}}", from=1-7, to=1-4]
	\arrow["{\alpha^2_{\mathcal{A}d_1}}", from=1-4, to=1-1]
	\arrow["{\alpha^1_{d_2}}", from=5-7, to=5-4]
	\arrow["{\alpha^2_{\mathcal{A}d_2}}", from=5-4, to=5-1]
	\arrow["{\eta_{d_1}}"{pos=0.7}, curve={height=30pt}, from=1-7, to=1-1]
	\arrow["{\eta_{d_2}}"{pos=0.7}, shift right=1, curve={height=30pt}, from=5-7, to=5-1]
	\arrow["{\iota_1}"{pos=0.3}, from=1-7, to=3-7]
	\arrow["{\iota_1}"'{pos=0.3}, from=5-7, to=3-7]
	\arrow["{\alpha^1_{d}}", from=3-7, to=3-4]
	\arrow["{\alpha^2_{\mathcal{A}d}}", from=3-4, to=3-1]
	\arrow["{\eta_d}"{pos=0.7}, curve={height=30pt}, from=3-7, to=3-1]
	\arrow["{\mathcal{A}_{\iota_1}}"{pos=0.3}, from=1-4, to=3-4]
	\arrow["{\mathcal{A}_{\iota_2}}"'{pos=0.3}, from=5-4, to=3-4]
	\arrow["{\mathcal{T}_{\iota_1}}"{pos=0.3}, from=1-1, to=3-1]
	\arrow["{\mathcal{T}_{\iota_2}}"'{pos=0.3}, from=5-1, to=3-1]
	\arrow["{\const e_{x_2}}"', from=3-9, to=1-7]
	\arrow["{\const e_{x_2}}", from=3-9, to=5-7]
\end{tikzcd}\]

These observations imply that the following is a pullback diagram in \(\finset\).
\begin{equation}\label{diagram:limd-is-a-pullback}
\begin{tikzcd}
	& {\mathcal{T}d_1(K_1) = \lim d_1} \\
	{\mathcal{T}d(K_1) = \lim d} && de \\
	& {\mathcal{T}d_2(K_1) =  \lim d_2}
	\arrow[from=1-2, to=2-3]
	\arrow[from=3-2, to=2-3]
	\arrow["{\rho_1}", from=2-1, to=1-2]
	\arrow["{\rho_2}"', from=2-1, to=3-2]
\end{tikzcd}
\end{equation}
We can factor this diagram further as follows (where \(\lambda_i\) is the leg of the cone with apex \(\lim d_i\)).
\[\begin{tikzcd}
	& {\mathcal{T}d_1(K_1) = \lim d_1} \\
	&& {\mathcal{A}d_1x_1} \\
	{\mathcal{T}d(K_1) = \lim d} &&& de \\
	&& {\mathcal{A}d_2x_2} \\
	& {\mathcal{T}d_2(K_1) =  \lim d_2}
	\arrow[curve={height=-24pt}, from=1-2, to=3-4]
	\arrow[curve={height=24pt}, from=5-2, to=3-4]
	\arrow["{\rho_1}", from=3-1, to=1-2]
	\arrow["{\rho_2}"', from=3-1, to=5-2]
	\arrow["{\lambda_1}", two heads, from=1-2, to=2-3]
	\arrow[from=2-3, to=3-4]
	\arrow["{\lambda_2}"', two heads, from=5-2, to=4-3]
	\arrow[from=4-3, to=3-4]
\end{tikzcd}\]
From which we observe that all that remains to be shown is that the unique pullback arrow \(u\) shown in the following diagram 
\[\begin{tikzcd}
	& {\mathcal{T}d_1(K_1) = \lim d_1} \\
	&& {\mathcal{A}d_1x_1} \\
	{\mathcal{T}d(K_1) = \lim d} & {\mathcal{A}d_1x_1\times_{de} \mathcal{A}d_2x_2} && de \\
	&& {\mathcal{A}d_2x_2} \\
	& {\mathcal{T}d_2(K_1) =  \lim d_2}
	\arrow[curve={height=-24pt}, from=1-2, to=3-4]
	\arrow[curve={height=24pt}, from=5-2, to=3-4]
	\arrow["{\rho_1}", from=3-1, to=1-2]
	\arrow["{\rho_2}"', from=3-1, to=5-2]
	\arrow["{\lambda_1}", two heads, from=1-2, to=2-3]
	\arrow[from=2-3, to=3-4]
	\arrow["{\lambda_2}"', two heads, from=5-2, to=4-3]
	\arrow[from=4-3, to=3-4]
	\arrow["{\pi_1}"', from=3-2, to=2-3]
	\arrow["{\pi_2}", from=3-2, to=4-3]
	\arrow["u", color={rgb,255:red,214;green,92;blue,92}, dashed, from=3-1, to=3-2]
\end{tikzcd}\]
is a surjection. To see why this suffices, notice that, if \(u\) is surjective, then the entire claim will follow since we would have 
\[
    \lambda_i \rho_i = \pi_i u \implies \image \lambda_i \rho_i = \image \pi_i u  = \image \pi_i|_{\image u} = \image \pi_i.
\]
But then this concludes the proof since the surjectivity of \(u\) is immediate once we recall that \(\mathcal{A}d_ix_i\) was defined as \(\image \lambda_i\) and that \(\lim d = \lim d_1 \times_{de} \lim d_2\) (as established in Diagram~\ref{diagram:limd-is-a-pullback}).
\end{proof}

Notice that Lemma~\ref{lemma:tree-algorithm-correctness} does not na\"ively lift to decompositions of arbitrary shapes. For instance consider the following example of a cyclic decomposition of a \(5\)-cycle graph with vertices \(\{x_1, \dots, x_5\}\).
\[ \begin{adjustbox}{scale=0.75}
    \begin{tikzcd}
	& {x_1} &&&& {x_2} \\
	&&& {x_1x_2} \\
	& {x_5x_1} &&&& {x_2x_3} \\
	&& {x_4x_5} && {x_3x_4} \\
	{x_5} &&& {x_4} &&& {x_3}
	\arrow[from=1-2, to=3-2]
	\arrow[from=1-2, to=2-4]
	\arrow[from=1-6, to=2-4]
	\arrow[from=1-6, to=3-6]
	\arrow[from=5-7, to=3-6]
	\arrow[from=5-7, to=4-5]
	\arrow[from=5-4, to=4-5]
	\arrow[from=5-4, to=4-3]
	\arrow[from=5-1, to=4-3]
	\arrow[from=5-1, to=3-2]
    \end{tikzcd}
\end{adjustbox} \]
Passing this decomposition through the two-coloring sheaf \(\cat{SimpFinGr}(-, K_2)\) (i.e. applying the functor \(\sd\cat{SimpFinGr}(-, K_2)\)) we obtain a structured \textit{co}-decomposition \(\delta\) valued in \(\finset\) and whose bags are all two element sets corresponding to the two proper colorings of any edge. Now notice that, since  odd cycles are not two-colorable at least one of the bags of \(\mathcal{A}\delta\) will be empty. In contrast, it is easy to verify that the output of  Algorithm~\ref{algorithm:single-edge-filtering} on \(\delta\) is isomorphic to \(\delta\) itself. 

Although these observations might seem to preclude us from obtaining algorithmic results on decompositions that are not tree-shaped, if we are willing to accept slower running times (which are still \(\fpt\)-time, but under a double parameterization rather than the single parameterization of Lemma~\ref{lemma:tree-algorithm-correctness}), then we can efficiently solve the sheaf decision problem on decompositions of other shapes as well. This is Theorem~\ref{thm:filtering-algorithm} which we are finally ready to prove. For clarity, we wish to point out that the following result is exactly the same as Lemma~\ref{lemma:tree-algorithm-correctness} when we are given a decomposition whose shape is a tree: trees trivially have feedback vertex number zero.

\begin{theorem}[Theorem~\ref{thm:filtering-algorithm} restated]
Let \(G\) be a finite, irreflexive, directed graph without antiparallel edges and at most one edge for each pair of vertices. Let \(\cat{D}\) be a small adhesively cocomplete category, let \(\Fa \colon \cat{D}^{op} \to \finset \) be a presheaf and let \(\cat{C}\) be one of \(\{\cat{D}, \:\cat{D}_{mono}\}\). If \(\Fa\) is a sheaf on the site \(\bigl(\cat{C}, \decomp|_{\cat{C}}\bigr)\) and if we are given an algorithm \(\Aa_{\Fa}\) which computes \(\Fa\) on any object \(c\) in time \(\alpha(c)\), then there is an algorithm which, given any \(\cat{C}\)-valued structured decomposition \(d \colon \int G \to \cat{C}\) of an object \(c \in \cat{C}\) and a feedback vertex set \(S\) of \(G\), computes \(\mathsf{dec} \Fa c \) in time \[\Oh(\max_{x \in VG}\alpha(dx) + \kappa^{|S|}\kappa^2)|EG|\] where \(\kappa = \max_{x \in VG}|\Fa dx|\). 
\end{theorem}
\begin{proof}[Proof of Theorem~\ref{thm:filtering-algorithm}]
Recall that, by Point~\ref{thm:A-is-functor-point-3} of Theorem~\ref{thm:A-is-functor} the following diagram commutes. 
\[  \]
Our proof will will rely on this fact together with the following claim. 
\begin{claim}\label{claim:thm:filtering-algorithm}
    Consider the image of \(d\) under \(\sd_{\Fa}\) and view it as a \(\finset\)-valued structured \textit{co}-decomposition \(\sd_{\Fa} d \colon (\int G)^{op} \to \finset\), fix any vertex \(x \in VG\) and any enumeration \(\ell = (e_1, \dots, e_n)\) of the edges incident with \(x\). Let \(\gamma_s\) be the output of Algorithm~\ref{algorithm:filtering} when applied to the input \((\sd_{\Fa} d, \ell)\) and, letting \(G' \hookrightarrow G\) be the subgraph of \(G\) obtained by removing all edges incident with \(x\), define \(\delta_s\) to be the decomposition \(\delta_s \colon (\int  G')^{op} \hookrightarrow (\int G)^{op} \xrightarrow{\gamma_s} \finset.\) \\ 
    If \(\sd_{\Fa}d(x)\) is a singleton, then \(\land \circ \mathsf{dec}^{op} \circ \Aa \circ \sd_{\Fa})(d) = (\land \circ \mathsf{dec}^{op} \circ \Aa)(\delta_s).\)
\end{claim}
\begin{proof}[Proof of Claim~\ref{claim:thm:filtering-algorithm}]
    If the bag \(\sd_{\Fa}dx\) has precisely one section, let's call it \(\xi\), then every matching family for \(\Fa\) on \(d\) must involve \(\xi\). Notice that we trivially have that \(\lim \sd_{\Fa}d = \lim \gamma_s\) since the local pullbacks performed by Algorithm~\ref{algorithm:single-edge-filtering} (and hence Algorithm~\ref{algorithm:filtering} cannot change the overall limit) and hence we have \[ (\land \circ \mathsf{dec}^{op} \circ \Aa \circ \sd_{\Fa})(d) = (\land \circ \mathsf{dec}^{op} \circ \Aa)(\gamma_s).\] But now notice that, since \(\sd_{\Fa}dx\) is a singleton, any collection of sections \((\zeta \in \delta_s(z))_{z \in VG'}\) must give rise to a matching family for \(d\) since, by the construction of \(\delta_s\), we have that each such section \(\zeta\) in the family must agree with \(\xi\). But then, as desired, we have proven that \[ (\land \circ \mathsf{dec}^{op} \circ \Aa \circ \sd_{\Fa})(d) = (\land \circ \mathsf{dec}^{op} \circ \Aa)(\delta_s).\]
\end{proof}
In light of this result, notice that, if \(\Fa ds\) is a singleton for each \(s \in S\) (where recall that \(S\) is a feedback vertex set in \(G\)), then, by repeatedly applying Claim~\ref{claim:thm:filtering-algorithm} until the output decomposition is a forest and then applying the algorithm of Lemma~\ref{lemma:tree-algorithm-correctness}, we can correctly solve the sheaf decision problem in time \(\Oh(\max_{x \in VG}\alpha(dx) + \kappa^2)|EG|\) (since this entire procedure amounts to one call to algorithm \(\Aa_{\Fa}\) in order to compute \(\sd_{\Fa} d\) and \(\Oh(|EG|)\)-calls to the edge-filtering algorithm -- i.e. Algorithm~\ref{algorithm:single-edge-filtering}).

Now, at an intuitive level, if there exists \(s \in S\) such that \(\Fa ds\) is not a singleton, then we can simply repeat the above procedure once for each section in \(\Fa ds\). Stating this more formally, define for each \(\sigma \in \prod_{s \in S} \Fa ds\) the \(\finset\)-valued structured co-decomposition \(\omega_\sigma \colon (\int G)^{op} \to \finset\) by replacing all bags of the form \(\sd_{\Fa} d s\) in \(\sd_{\Fa} d\) with the bag \(\omega_\sigma s := \{\sigma_s\}\) (where \(\sigma_s\) is the section at index \(s\) in the tuple \(\sigma\)). Then, since \(\Fa\) is a sheaf, it follows that \[\mathsf{dec}^{op}\Fa c = \bigvee_{\sigma \in \prod_{s \in S} \Fa ds} (\land \circ \sd \mathsf{dec}^{op} \circ \mathcal{A})(\omega_\sigma).\] 
This will have the desired running time since it corresponds to first computing \(\sd_{\Fa} d\) (which we can do in time \(\Oh(\max_{x \in VG}\alpha(dx))\) using algorithm \(\Aa_{\Fa}\)) and then running Algorithm~\ref{algorithm:filtering} (whose correctness and running time are established by Claim~\ref{claim:thm:filtering-algorithm} and Lemma~\ref{lemma:tree-algorithm-correctness}) at most \((|\prod_{s \in S} \Fa ds| \in \Oh( k^{|S|} ))\)-many times.
\end{proof}

The reader might notice that we have so far avoided mentioning notions of \textit{width} of the input decompositions and that indeed these considerations do not appear in the statements of the algorithms of Lemma~\ref{lemma:tree-algorithm-correctness} or Theorem~\ref{thm:filtering-algorithm}. This is for good reason: one should observe that it is not the width of the decompositions of the inputs that matters; instead it is the \textit{width of the decompositions of the solutions spaces} that is key to the algorithmic bounds. In categories (such as that of graphs, say) where objects come equipped with natural notions of `size', then one might expect that it would be convenient to state the running time of the algorithm in terms of the maximum bag size in the input decomposition. However, we maintain (in accordance with observations previously made by Bodlaender and Fomin~\cite{bodlaender2002tree}) that quantifying the running time of the algorithm in terms of the width of the decompositions of the inputs is misleading. To see why, consider the very simple, but concrete case of algorithms for coloring on tree decompositions. In this situation, it is perfectly fine to admit very large bags in our decomposition so long as the following two conditions are met: \begin{enumerate*}
    \item the local solution spaces (which are sets) associated to these bags are small and
    \item these solution spaces can be determined in a bounded amount of time.
\end{enumerate*} A trivial, but illuminating example of this phenomenon is the case in which we allow large bags (of unbounded size) in our decomposition as long as they consist of complete graphs: for such graphs we have only one proper coloring up to isomorphism and this can be determined in linear time.

\subsection{Implementation}\label{sec:implementation}
Compared to the traditional, combinatorial definition of graph decompositions~\cite{Bertele1972NonserialProgramming, halin1976s, RobertsonII}, our category theoretic formulation of structured decompositions has two advantages which we have already encountered.
\begin{enumerate}
    \item \textit{Object agnosticism} Structured decompositions allow us to describe decompositions of objects of any adhesive category and thus one doesn't need to define ad-hoc decompositions on a case-by-case basis whenever one encounters a new kind of combinatorial data.
    \item \textit{Functorial Algorithmics} The functoriality of \(\sd\) (i.e. of categories of structured decompositions) allows us to make explicit use of solution spaces and decompositions thereof. This allows us to state the correctness of algorithmic results as the commutativity of appropriate diagrams (e.g. Diagram~\eqref{diagram:algorithmic-goal}) from which one can moreover infer running time bottlenecks. 
\end{enumerate}
However, there is a further, more practical benefit of our category-theoretic perspective: it allows for a very smooth transition from mathematics to implementation. Indeed, our theoretical algorithmic results of Section~\ref{sec:deciding-sheaves} can be easily paired with corresponding implementations~\cite{str_decomps_julia_implementation} in the AlgebraicJulia ecosystem~\cite{evan_patterson_algebraicjuliacatlabjl_2020}. As a proof of concept, we have implemented structured decompositions and Algorithm~\ref{algorithm:filtering} for \(\textsc{SheafDecision}\) on tree-shaped decompositions which demonstrates the seamless transition from mathematics to code which one can experience one category theory is embraced as the core theoretical abstraction. We encourage the reader to consult the relevant repository~\cite{str_decomps_julia_implementation} for further details of the implementation.

\section{An Open Problem on the Shapes of the Decompositions}
Theorem~\ref{thm:filtering-algorithm} yields \(\fpt\)-time algorithms for problems encoded as sheaves on adhesive categories. When instantiated, this allows us to obtain algorithmic results on many mathematical objects of algorithmic interest such as: \begin{enumerate*} \item databases,\item simple graphs, \item directed graphs, \item directed multigraphs, \item hypergraphs, \item directed hypergraphs, \item simplicial complexes, \item circular port graphs~\cite{libkind2021operadic} and \item half-edge graphs \end{enumerate*}. However one should note that these parameterizations are only useful if: (1) not all objects in this class have bounded structured decomposition width and (2) only if the feedback vertex number of the class of decomposition shapes is bounded. Notice that it is easy to verify that, for all the examples mentioned above, not all objects have bounded \textit{tree-shaped} decompositions. However, when it comes to decomposition shapes that are not trees, the question of whether all objects have bounded width with respect to a fixed class of decomposition shapes, it not obvious. Indeed, this will motivate Open Problem~\ref{problem:shape-graph-classes} which will arise naturally by the end of this section in which we will enquire about how our algorithmic results -- which involve arbitrary decomposition shapes -- compare to more traditional results in graph theory which solely make use of tree-shaped decompositions. In particular we will now briefly argue that, for the case of \textit{graph decomposition width} (as defined by Carmesin~\cite{CARMESIN2022101}), our results end up yielding \(\fpt\) algorithms only on classes of graphs which have bounded tree-width. Determining whether these observations can be carried over to more general classes of objects is a fascinating new direction for work (see Open Problem~\ref{problem:shape-graph-classes}).

To see this, first of all note that, for any \(\fingr\)-valued decomposition \(d: \int H \to \fingr \) of a graph \(G\), it is easy to see that the treewidth $\tw(H)$ of our shape graph $H$ is at most its feedback vertex number. 
Furthermore, if the shape graph $H$ has treewidth at most $t_H$ and the input graph $G$ has $H$-width at most $t_G$, then we can easily build a tree decomposition of $G$ of width $t_H \cdot t_G$, implying that we only compute on graphs with bounded treewidth.
Since we are now only dealing with graphs, it will be convenient to switch to Carmesin's~\cite{CARMESIN2022101} more combinatorial notation\footnote{We refer the reader to Bumpus, Kocis and Master~\cite{structured-decompositions} for details of how to choose a graph-theoretic instantiation of our categorical notation which neatly corresponds to that of Carmesin's notation of graph decompositions~\cite{CARMESIN2022101}.}. Let $\Ta^T=(T,(X_t)_{t \in V(T)})$ be a tree decomposition of $H$ and $\Ta^G=(H,(Y_t)_{t \in V(H)})$ be a $H$-decomposition of $G$. We claim that $\Ta=(T, (\cup_{t' \in X_t} Y_{t'} )_{t \in V(T)}$ is a tree decomposition of $G$. Coverage is easy as each bag of $\Ta^H$ is contained in at least one bag of $\Ta$ as all bags are covered in $\Ta^G$. For the coherence, we argue as follows. Assume a vertex $v \in V(G)$ is in two bags $t, t' \in V(\Ta)$, but not in a bag $t''$ on the path in $\Ta$ from $t$ to $t'$. By the construction of the bags of $\Ta$, there is a vertex $u \in Y_t$ whose bag contains $v$ and similarly a vertex $u' \in Y_{t'}$ containing $v$. By the coherence of $\Ta^H$, there has to be a path $p$ between $u$ and $u'$ in $\Ta^H$ such that $v$ is in all bags $Y_{\tilde t}$ for $\tilde t \in V(p)$ (the subgraph of $H$ induced by the bags containing $v$ is connected).  As $v$ is not in the bag of $t''$ of $\Ta$, no vertex of $V(p)$ is in $X_{t''}$ (of $\Ta^T$), which contradicts that the removal of $X_{t''}$ separates $u \in X_t \setminus X_{t''}$ from $X_{t'} \setminus X_{t''} \ni u'$.

Hence, if the treewidth of the graphs of the class $\mathcal{H}$ allowed for the shape of the structured decomposition is bounded, the graphs with bounded $\mathcal{H}$-width have bounded treewidth. Hence, to obtain new FPT algorithms the treewdith of the graphs of the class $\mathcal{H}$ should not be bounded. On the other hand, the class $\mathcal{H}$ has to be quite restricted, as otherwise each graph will have bounded $\mathcal{H}$-width. For example, if $\mathcal{H}$ contains the $n \times n$-grid, every $n$-vertex graph was width $1$ in the class, as we can see as follows. Consider the $n \times n$-grid $H$ and let $v_{i,j}$ for $1 \le i,j \le n$ be such that $v_{i-1,j}v_{i,j} \in E(H)$ for $2 \le i \le n$ and $v_{i,j-1}v_{i,j} \in E(H)$ for $2 \le j \le n$.  We claim that $(H, (\{ i,j \})_{v_{i,j} \in V(H)})$ is a $H$-decomposition of every graph with vertex set $\{1, \dots, n \}$. Coverage is easy to see as for each $1 \le i,j \le n$, we constructed a bag containing $i$ and $j$ and hence the complete graph over $\{1, \dots, n \}$ is covered. For the coherence, notice that a vertex $k$ is contained in the bags $X_{k,j}$ for $1 \le j \le n$ and $X_{i,k}$ for $1 \le i \le n$, which build a cross in the grid $H$ and hence these bags are connected.

This construction can be generalized to graphs having the $n \times n$-grid as a minor by making the bag of a vertex contracted with $v_{i,j}$ equal to the bag $X_{v_{i,j}}$ and making all bags of vertices that are removed empty. Hence, all graphs have planar-width at most one. This motivates the following fascinating open problem. 
\begin{problem}\label{problem:shape-graph-classes}
Does there exists an adhesive category \(\cat{C}\) such that there is a class \(\chi\) of objects in \(\cat{C}\) and a graph class \(\Ga\) of bounded feedback vertex number such that the following two conditions hold simultaneously?  
\begin{itemize}
    \item The class \(\chi\) has \textbf{unbounded} \textit{tree-shaped} structured decomposition width and
    \item the class \(\chi\) has \textbf{bounded} \(\Ga\)-shaped structured decomposition width.
\end{itemize}
\end{problem}

\section{Discussion}\label{sec:discussion}
Our main contribution is to bridge the ``\textbf{structure}'' and ``\textbf{power}'' communities by proving an algorithmic meta-theorem (Theorem~\ref{thm:filtering-algorithm}) which informally states that decision problems encoded as sheaves (\textbf{Representational Compositionality}) can be solved by dynamic programming (\textbf{Algorithmic Compositionality}) in linear time on classes of inputs which admit structured decompositions of bounded width and whose decomposition shape has bounded feedback vertex number (\textbf{Structural Compositionality}). Our results thus bridge the mathematical and linguistic differences of these two communities -- of ``structure'' and ``power'' --  by showing how to use category theory and sheaf theory to amalgamate three kinds of compositionality found in mathematics and theoretical computer science. This is summarized at a very high level via the following diagram (i.e.  Diagram~\ref{diagram:ENGLISH} which was formalized as Diagram~\ref{diagram:algorithmic-goal-NO-COLIMITS}). 

\adjustbox{scale=1.5, max width=\textwidth}{

}

\paragraph{Future work.} Other than Open Problem~\ref{problem:shape-graph-classes}, directions for further work abound. Here we mention but a few obvious, yet exciting candidates. First of all it is clear that, although our meta-theoretic results achieve a remarkable degree of \textit{horizontal} generality (in terms of the kinds of mathematical structures to which they apply), their \textit{vertical} generality (the breadth of problems which can be solved) is still surpassed by more traditional results such as Courcelle's theorem~\cite{CourcelleBook}. It is a fascinating direction for further work to understand the connection between these model-theoretic approaches and our category- and sheaf-theoretic ones. Furthermore, we provide the following two more concrete lines of future work.
\begin{enumerate}
    \item Studying the connections between other kinds of topologies (such as those generated by more permissive containment relations such as graph minors) and the topologies given by structured decompositions.
    \item For other problems -- such as vertex cover or Hamilton Path -- we need more tools since, although they can be presented as presheaves, they fail to be sheaves. Indeed, note that this failure of compositionality ``on the nose'' is not specific to our approach and it exists in algorithmics as well: when solving Hamilton path by dynamic programming on a tree decomposition, one does not map the bags of the decomposition to local Hamilton Paths, but instead to disjoint collections of paths (see Flum and Grohe~\cite{flum2006parameterized} for details). The seasoned algorithmicist will point out that there is a template which can be often followed in order to choose the correct partial solutions: consider a global solution, induce it locally on the bags and then use this information to determine the obstructions to algorithmic compositionality. From our perspective this looks a lot like asking: ``what are the obstructions to the problem being a sheaf and how can we systematically track that information?'' Fortunately there are many powerful and well-developed tools from sheaf cohomology which appear to be appropriate for this task. This is an exciting new direction for research which we are already actively exploring. 
\end{enumerate}

\bibliographystyle{fundam}
\bibliography{biblio}


\appendix
\section{Notions from Sheaf Theory}\label{appendix:sheaf-theory}
Here we collect basic sheaf-theoretic definitions which we use throughout the document. We refer the reader to Rosiak's textbook~\cite{rosiak-book} for an introduction to sheaf theory which is suitable for beginners.  

\begin{definition}\label{def:pre-topology}
    Let \(\cat{C}\) be a category with pullbacks and let \(K\) be a function assigning to each object \(c\) in \(\cat{C}\) a family of sets of morphisms with codomain \(c\) called a \define{family of covering sets}. We call \(K\) a \define{a Grothendieck pre-topology on \(\cat{C}\) } if the following three conditions hold for all \(c \in C\).
    \begin{enumerate}[label=\textbf{( PT\arabic*)}]
        \item If \(f : c' \to c\) is an isomorphism, then \(\{f\} \in K(c)\).\label{axiom:pretopology-1}
        \item If \(S \in K(c)\) and \(g: b \to c\) is a morphism in \(C\), then the pullback of \(S\) and \(g\) is a covering set for \(b\); i.e. \(\{g \times_{c} f \mid f \in S\} \in K(b).\)\label{axiom:pretopology-2}
        \item If \(\{f_i : c_i \to c \mid i \in I\} \in K(c)\), then, whenever we are given \[\{g_{ij}: b_{ij} \to c_i | j \in J_i\} \in K(c_i)\] for all \(i \in I\), we must have \[\{f_i \circ g_{ij}: b_{ij} \to c_i \to c \mid i \in I, j \in J_i\} \in K(c).\]\label{axiom:pretopology-3}
    \end{enumerate}
\end{definition}

\begin{definition}\label{def:sieve}
A \define{sieve} on an object \(c \in \cat{C}\) is a family \(S_c\) of morphisms with codomain \(c\) which is closed under pre-composition.   
\end{definition}
Sieves are essentially what happens when we decide to just use those families that are ``saturated'' (in the sense that they are closed under pre-composition with morphisms in $\cat{C}$); they are introduced in order to present the definition of a Grothendieck topology (see \ref{def:Grothendieck-topology}), a revision of \ref{def:pre-topology}.

Observe that (at least when $\cat{C}$ is locally small) sieves can be identified with subfunctors of the representable hom-functor $\hom (\--, c) = y_c$ (note that the category $\mathbf{Set}^{\cat{C}^{op}}$ of presheaves on $\cat{C}$ has pullbacks, letting us drop, via this sieve approach, the assumption on $\cat{C}$ itself). Note also that if $S \subseteq \hom (\--, c)$ is a sieve on $c$ and $f: b \rightarrow c$ is any morphism to $c$, then the pullback sieve along $f$ 
\begin{equation*}
    f^*(S) = \{g | \codomain(g) = b, f \circ g \in S \} 
\end{equation*}
is a sieve on $b$. 

Typically, the definition of a Grothendieck topology is then given in terms of a function $J$ assigning sieves to the objects of $\cat{C}$ such that three axioms (maximal sieve, stability under base change, and transitivity) are satisfied. But really the stability axiom -- requiring that if $S \in J(c)$, then $f^*(S) \in J(b)$ for any arrow $f: b \rightarrow c$ -- just amounts to requiring that $J$ is in fact a \textit{functor} $J: \cat{C}^{op} \rightarrow \mathbf{Set}$, i.e., an object in the presheaf category $\mathbf{Set}^{\cat{C}^{op}}$. If we let \(\sieve : \cat{C}^{op} \to \mathbf{Set}\) be the functor which takes objects \(c\) of a small category \(\cat{C}\) to the set of all sieves on \(c\), and for any morphism $f : x \to y$ in \(\cat{C}\) use the map \[f^* \colon (S \; \in \sieve(y)) \mapsto \bigl( \{g \colon w \to x \mid (fg: w \to x \to y) \; \in S\} \in \sieve(x) \bigr),\] then we can give a somewhat condensed definition of a Grothendieck topology (by already building the usual stability axiom into the characterization of $J$ as a particular functor), as follows.

\begin{definition}\label{def:Grothendieck-topology}
Let \(\cat{C}\) be a category and consider a subfunctor \(J\) of \(\sieve: \cat{C}^{op} \to \mathbf{Set}\) of "permissible sieves". We call \(J\) a \define{Grothendieck topology on \(\cat{C}\)} if it satisfies the following two conditions: 
\begin{enumerate}[label=\textbf{( Gt\arabic*)}]
    \item the maximal sieve is always permissible; i.e. \(\{f \in \mathsf{Mor}(\cat{C}) \mid \codomain(f) = c \} \in J(c)\) for all \(c \in \cat{C}\) \label{axiom:GT1}
    \item for any sieve \(R\) on some object \(c \in \cat{C}\), if there is a permissible sieve \(S \in J(c)\) on \(c\) such that for all \((h: b \to c) \in S\) the sieve \(h^*(R)\) is permissible on \(b\) (i.e. \(h^*(R) \in J(b)\)), then \(R\) is itself permissible on \(c\) (i.e. \(R \in J(c)\)). \label{axiom:GT2}
\end{enumerate}
For any \(c \in \cat{C}\) we call the elements \(S \in J(c)\) \define{\(J\)-covers} or simply \define{covers}, if \(J\) is understood from context. A \define{site} is a pair \((\cat{C}, J)\) consisting of a category \(\cat{C}\) and a Grothendieck topology on \(\cat{C}\).
\end{definition}

\begin{definition}\label{def:matching-family}
    Let \((\cat{C}, J)\) be a site, \(S\) be a \(J\)-cover of an object \(c \in \cat{C}\) and \(P: \cat{C}^{op} \to \mathbf{Set}\) be a presheaf. Then a \define{matching family} of sections of \(P\) with respect to the cover \(S\) is a morphism (natural transformation) of presheaves \(\chi \colon S \Rightarrow P\).
\end{definition}

\begin{definition}\label{def:sheaf}
    Let \(P \colon \cat{C}^{op} \to \mathbf{Set}\) be a presheaf on a site \((\cat{C}, J)\). Then we call \(P\) a \define{sheaf with respect to \(J\) }(or a \(J\)-sheaf) if for all \(c \in \cat{C}\) and for all covering sieves \(\bigl(\iota \colon S \Rightarrow y_c\bigr) \in J(c)\) of \(c\) each matching family \( \chi: S \Rightarrow P\) has a unique extension to the morphism \(\mathcal{E} \colon y_c \Rightarrow P\).
\end{definition}

\end{document}